\documentclass{article}

\usepackage{microtype}
\usepackage{graphicx}
\usepackage{subfigure}
\usepackage{booktabs} 

\usepackage{hyperref}

\usepackage[accepted]{icml2024} 

\usepackage{amsmath}
\usepackage{amssymb}
\usepackage{mathtools}
\usepackage{amsthm}

\usepackage[capitalize,noabbrev]{cleveref}

\theoremstyle{plain}
\newtheorem{thm}{Theorem}

\theoremstyle{definition}

\theoremstyle{remark}

\usepackage{amsfonts,amscd,mathrsfs}
\usepackage{thm-restate}
\usepackage{natbib}
\usepackage{multirow}

\icmltitlerunning{A Fast Algorithm to Simulate Nonlinear Resistive Networks}

\begin{document}

\twocolumn[
\icmltitle{A Fast Algorithm to Simulate Nonlinear Resistive Networks}

\begin{icmlauthorlist}
\icmlauthor{Benjamin Scellier}{rain}
\end{icmlauthorlist}

\icmlaffiliation{rain}{Rain AI, San Francisco, CA, USA}

\icmlcorrespondingauthor{Benjamin Scellier}{benjamin@rain.ai}

\icmlkeywords{resistor network,nonlinear resistive network,deep resistive network,convex quadratic programming,block coordinate descent,self-learning machine,equilibrium propagation}

\vskip 0.3in
]


\printAffiliationsAndNotice{}

\begin{abstract}
Analog electrical networks have long been investigated as energy-efficient computing platforms for machine learning, leveraging analog physics during inference. More recently, resistor networks have sparked particular interest due to their ability to learn using local rules (such as equilibrium propagation), enabling potentially important energy efficiency gains for training as well. Despite their potential advantage, the simulations of these resistor networks has been a significant bottleneck to assess their scalability, with current methods either being limited to linear networks or relying on realistic, yet slow circuit simulators like SPICE. Assuming ideal circuit elements, we introduce a novel approach for the simulation of nonlinear resistive networks, which we frame as a quadratic programming problem with linear inequality constraints, and which we solve using a fast, exact coordinate descent algorithm. Our simulation methodology significantly outperforms existing SPICE-based simulations, enabling the training of networks up to 327 times larger at speeds 160 times faster, resulting in a 50,000-fold improvement in the ratio of network size to epoch duration. Our approach can foster more rapid progress in the simulations of nonlinear analog electrical networks.
\end{abstract}

\section{Introduction}

As energy costs associated with machine learning are rapidly increasing, neuromorphic computing platforms are being explored as alternatives to neural networks and GPUs \citep{markovic2020physics}. Such platforms leverage analog physics and compute-in-memory architectures to achieve significant energy gains.
In particular,
nonlinear resistive networks have recently sparked interest \citep{kendall2020training,dillavou2022demonstration,dillavou2023machine,wycoff2022desynchronous,anisetti2024frequency,stern2022physical,stern2024training,kiraz2022impacts,watfa2023energy,oh2023memristor}. Central to these networks are variable resistors (e.g. memristors) which act as trainable weights, coupled with diodes that introduce nonlinearities, as well as voltage and current sources for input signals. The conductances of these variable resistors can be adjusted, allowing the network to be trained to achieve specific computational tasks. Notably, like traditional neural networks, nonlinear resistive networks have two essential properties: they are universal function approximators \citep{scellier2023universal} and they can be trained by gradient descent \citep{kendall2020training,anisetti2024frequency}. Unlike GPU-based neural networks, however, they leverage the laws of electrical circuits (e.g. Kirchhoff's laws and Ohm's law) to perform inference and extract the weight gradients. Furthermore, the learning rules governing conductance changes are local. These features make nonlinear resistive networks good candidates as power-efficient learning-capable hardware, with recent experiments on memristive networks suggesting a potential 10,000x gain in energy efficiency compared to neural networks trained on GPUs \citep{yi2023activity}. Small-scale self-learning variable resistor networks have been built and succesfully trained on datasets such as Iris \citep{dillavou2022demonstration,dillavou2023machine}, validating the soundness of the approach.

Simulations of larger nonlinear resistive networks on tasks such as MNIST \citep{kendall2020training} and Fashion-MNIST \citep{watfa2023energy} further underscore their potential. These works and others \citep{kiraz2022impacts,oh2023memristor} performed the simulations using the general-purpose SPICE circuit simulator \citep{keiter2014xyce,vogt2020ngspice}. However, these efforts have been hampered by the slowness of SPICE, which is not specifically conceived to perform efficient simulations of resistive networks used for machine learning applications. To illustrate, \citet{kendall2020training} employed SPICE to simulate the training on MNIST of a one-hidden-layer network comprising 100 hidden nodes, a process which took one week for only ten epochs of training. Due to a scarcity of methods and algorithms for nonlinear resistive network simulations, another line of works resorts to linear networks \citep{stern2022physical,stern2024training,wycoff2022desynchronous}, thus missing a fundamental ingredient of modern machine learning applications: nonlinearity. In order to better understand nonlinear resistive networks, guide their hardware design, and further demonstrate their scalability to more complex tasks, the ability to efficiently simulate them has thus become crucial.

We introduce a novel methodology tailored for the simulations of nonlinear resistive networks. Our algorithm, applicable to networks with arbitrary topologies, is an instance of an exact coordinate descent algorithm for convex quadratic programming (QP) problems with linear constraints \citep{wright2015coordinate}. When applied to the `deep resistive network' architecture of \citet{kendall2020training}, our algorithm is an instance of an `exact block coordinate descent' algorithm, whose run-time on GPUs is orders of magnitude faster than SPICE. The contributions of the present manuscript are the following:
\begin{itemize}
\item We show that, in a nonlinear resistive network, under an assumption of ideality of the circuit elements (resistors, diodes, voltage sources and current sources), the steady state configuration of node electrical potentials is the solution of a convex minimization problem: specifically, a convex quadratic programming (QP) problem with linear inequality constraints (Theorem~\ref{thm:convex-qp-formulation}).
\item Using the QP formulation, we derive an algorithm to compute the steady state of an ideal nonlinear resistive network (Theorem~\ref{thm:coordinate-descent}). Our algorithm, which is an instance of an `exact coordinate descent' algorithm, is applicable to networks with arbitrary topologies.
\item For a specific class of nonlinear resistive networks called `deep resistive networks' (DRNs), we derive a specialized, fast, algorithm to compute the steady state (Section~\ref{sec:deep-resistive-network}). It exploits the bipartite structure of the DRN to perform exact block coordinate descent, where half of the coordinates (node electrical potentials) are updated in parallel at each step of the minimization process. Each step of our algorithm involves solely tensor multiplications, divisions, and clipping, making it amenable to fast executions on parallel computers such as GPUs.
\item We perform simulations of ideal DRNs, trained with equilibrium propagation (EP) on the MNIST dataset (Section \ref{sec:simulations}). Compared to the SPICE-based simulations of \citet{kendall2020training}, our DRNs have up to 327 times more parameters (variable resistors) and the training time per epoch is 160 shorter, resulting in a 50000x larger network size to epoch duration ratio. We also train DRNs of two and three hidden layers for the first time.
\item We compare in simulations DRNs with their closely related deep Hopfield networks (DHNs) studied in \citet{scellier2017equilibrium} (Appendix~\ref{sec:continuous-hopfield-networks}). We show that, while both models are comparable in performance, DRNs require much fewer iterations to converge to steady state (e.g. 6 iterations for a 3-hidden-layer DRN vs 100 iterations for a 3-hidden-layer DHN).
\item We prove that the equilibrium propagation formulas of \citet{scellier2017equilibrium,scellier2023energy} remain valid in the setting of ideal nonlinear resistive networks, and more generally when the set of feasible configurations is defined by inequality constraints (Appendix~\ref{sec:ep-inequality}).
\end{itemize}

We emphasize that our methodology to simulate nonlinear resistive networks relies on an assumption of ideality of the circuit elements — an abstraction that might appear overly simplistic. Although the real-world deviations from ideality are acknowledged, the great acceleration in simulation times that our methodology offers opens the door to large-scale simulations of nonlinear resistive networks, which we believe may also be informative about the behaviour of real (non-ideal) networks.

\section{Nonlinear Resistive Networks}
\label{sec:nonlinear-resistive-network}

In this work, we study electrical circuits composed of voltage sources, current sources, linear resistors and diodes. We refer to these circuits as nonlinear resistive networks.

\subsection{Model and Assumptions}
\label{sec:ideality-assumptions}

Following \citet{scellier2023universal}, we assume that the four circuit elements are \textit{ideal}, with their behaviour determined by the following current-voltage ($i$-$v$) characteristics (Figure~\ref{fig:resistive-elements}):
\begin{itemize}
\item Voltage Source: Satisfies $v=v_0$ for a constant voltage $v_0$, regardless of the current $i$.
\item Current Source: Satisfies $i = i_0$ for a constant current $i_0$, regardless of the voltage $v$. 
\item Linear Resistor: Follows Ohm's law, $i = g v$, where $g$ is the conductance ($g=1/r$, with $r$ being the resistance).
\item Diode: Satisfies $i=0$ for $v < 0$, and $v=0$ for $i>0$.
\end{itemize}
An ideal diode thus has two states: the off-state, where it behaves like an open switch, allowing no current to flow regardless of the voltage drop across it, or the on-state, where it behaves like a closed switch, allowing current to flow without any voltage drop.

\begin{figure*}
\begin{center}
\includegraphics[width=0.9\textwidth]{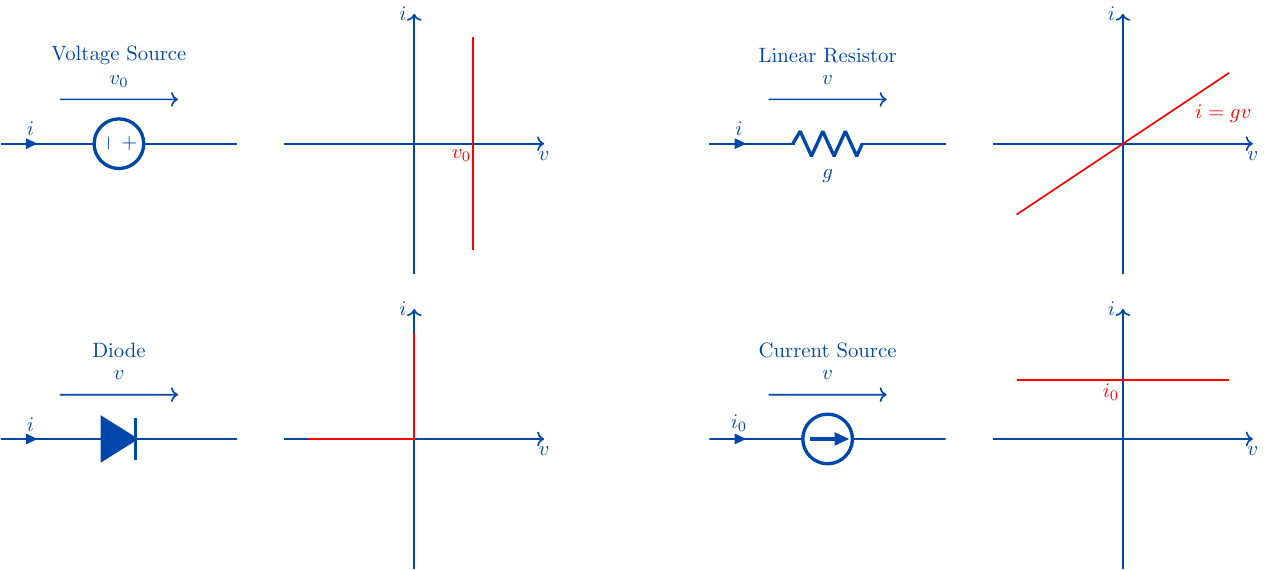}
\end{center}
\caption{
\textbf{Ideal circuit elements and their current-voltage (i-v) characteristics.} A linear resistor follows Ohm's law: $i = g v$, where $g$ is the conductance ($g=1/r$, with $r$ being the resistance). An ideal diode is characterized by $i=0$ for $v \leq 0$ ("off-state") and $v=0$ for $i>0$ ("on-state"). An ideal voltage source is characterized by $v=v_0$ for a constant voltage $v_0$ independent of the current $i$. An ideal current source is characterized by $i = i_0$ for a constant current $i_0$ independent of the voltage $v$.
}
\label{fig:resistive-elements}
\end{figure*}

\begin{figure*}
\begin{center}
\includegraphics[width=0.6\textwidth]{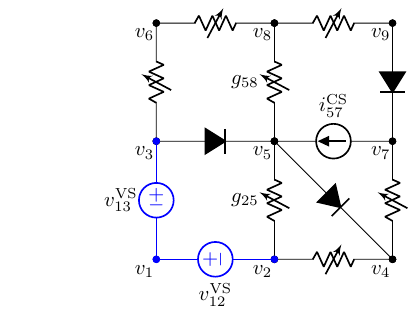}
\end{center}
\caption{
\textbf{A nonlinear resistive network.} By assumption, the voltage sources form a tree (in blue), so if we set e.g. $v_1=0$, we can immediately infer $v_2 = -v_{12}^{\rm VS}$ and $v_3 = v_{13}^{\rm VS}$. Next, we compute the steady state of the network by performing exact coordinate descent (Theorem~\ref{thm:coordinate-descent}) on the set of internal node electrical potentials (in black). As an example, one step of exact coordinate descent on node $k=5$ proceeds as follows. First we look at the resistors and current sources connected to node $k=5$ and we calculate $p_5 = (g_{25} v_2 + g_{58} v_8 + i_{57}^{\rm CS}) / (g_{25} + g_{58})$. Then we look at the diodes connected to node $k=5$ and we calculate $v_5 = \max(v_3, \min(p_5, v_4))$. This is the value of $v_5$ that achieves the minimum of $E(v_5)$ given other variables (node electrical potentials) fixed. We repeat the process with other nodes until convergence.
\label{fig:nonlinear-resistive-network}
}
\end{figure*}

A nonlinear resistive network can be represented as a graph where each branch contains a unique element, as shown in Figure~\ref{fig:nonlinear-resistive-network}. We denote the set of all branches as $\mathcal{B} = \mathcal{B}_{\rm VS} \cup \mathcal{B}_{\rm CS} \cup \mathcal{B}_{\rm R} \cup \mathcal{B}_{\rm D}$, where $\mathcal{B}_{\rm VS}$, $\mathcal{B}_{\rm CS}$, $\mathcal{B}_{\rm R}$ and $\mathcal{B}_{\rm D}$ are the subsets containing voltage sources, current sources, resistors and diodes, respectively. For every $(j,k) \in \mathcal{B}_{\rm VS}$, we denote the voltage across the voltage source between nodes $j$ and $k$ as $v_{jk}^{\rm VS}$. For every $(j,k) \in \mathcal{B}_{\rm CS}$, we denote the current through the current source between nodes $j$ and $k$ as $i_{jk}^{\rm CS}$. Finally, for every $(j,k) \in \mathcal{B}_{\rm R}$, we denote the conductance of the resistor between nodes $j$ and $k$ as $g_{jk}$.

We define a `steady state' as a configuration of branch voltages and branch currents that satisfies the above branch equations, as well as Kirchhoff's current law (KCL) at every node, and Kirchhoff's voltage law (KVL) in every loop. Under the above assumption of ideality, the steady state of a nonlinear resistive network is characterized by the following result -- see Appendix~\ref{sec:proofs} for a proof.

\begin{restatable}[Convex QP formulation]{thm}{convexqpformulation}
\label{thm:convex-qp-formulation}
Consider a nonlinear resistive network with $N$ nodes, and denote $v = (v_1,v_2,\ldots,v_N)$ the vector of node electrical potentials\footnote{The electrical potentials are defined up to a constant, so we may assume, for instance, $v_1=0$.}. Under the assumption of ideality, the steady state configuration of node electrical potentials, denoted $v_\star$, satisfies
\begin{equation}
\label{eq:energy-minimum}
v_\star = \underset{v \in \mathcal{S}}{\arg \min} \; E(v),
\end{equation}
where $E: \mathbb{R}^N \to \mathbb{R}$ is defined by
\begin{align}
\label{eq:energy-function}
E(v_1, \ldots, v_N) & := \frac{1}{2} \sum_{(j,k) \in \mathcal{B}_{\rm R}} g_{jk} \left( v_j - v_k \right)^2 \\
& + \sum_{(j,k) \in \mathcal{B}_{\rm CS}} i_{jk}^{\rm CS} \left( v_j - v_k \right),
\end{align}
and $\mathcal{S}$ is defined as:
\begin{align}
\label{eq:feasible-set}
\mathcal{S} := \{ & (v_1,v_2,\ldots,v_N) \in \mathbb{R}^N, \\
& v_j \leq v_k \quad \forall (j,k) \in \mathcal{B}_{\rm D}, \\
& v_j = v_k + v_{jk}^{\rm VS} \quad \forall (j,k) \in \mathcal{B}_{\rm VS} \}.
\end{align}
\end{restatable}

A few comments are in order. First, $\mathcal{S}$ is the set of feasible configurations of node electrical potentials (or the \textit{feasible set} for short), and $E$ is the energy function of the network (or the \textit{objective function}). Importantly, a configuration $v \in \mathcal{S}$ does not necessarily satisfy all the laws of electrical circuit theory - in other words, a \textit{feasible} configuration is not necessarily the \textit{physically realized} configuration. Similarly, the energy function $E$ is defined for every feasible configuration $v \in \mathcal{S}$, even those that do not comply with all the laws of electrical circuit theory. Theorem~\ref{thm:convex-qp-formulation} states that among all feasible configurations, the one that is physically realized (the steady state) is the configuration that minimizes $E$.

Second, the feasible set $\mathcal{S}$ is defined by linear equality and inequality constraints. The constraint $v_j \geq v_k$ for every $(j,k) \in \mathcal{B}_{\rm D}$ ensures the voltage $v_j-v_k$ across the diode is non-negative ($v_j > v_k$ if the diode is in the off-state, and $v_j = v_k$ if the diode is in the on-state). The constraint $v_j = v_k + v_{jk}^{\rm VS}$ for every $(j,k) \in \mathcal{B}_{\rm VS}$ ensures the voltage $v_j-v_k$ across the voltage source is $v_{jk}^{\rm VS}$. The diodes and voltage sources thus constrain the set of feasible configurations. We note that some conditions on the network topology and branch characteristics must be met to ensure that the feasible set $\mathcal{S}$ is non-empty: for instance, if the network contains a loop of voltage sources whose voltage drops do not sum to zero, KVL is violated and the feasible set is empty.

Third, the energy function $E$ is half the total power dissipated in the resistors plus the total power dissipated in the current sources. Specifically, $g_{jk} \left( v_j - v_k \right)^2$ represents the power dissipated in the resistor of branch $(j,k) \in \mathcal{B}_{\rm R}$, and $i_{jk}^{\rm CS} \left( v_j - v_k \right)$ represents the power dissipated in the current source of branch $(j,k) \in \mathcal{B}_{\rm CS}$. Theorem \ref{thm:convex-qp-formulation} generalizes the \textit{principle of minimum dissipated power} which states that, in a linear resistor network (with voltage sources and linear resistors, but without diodes and current sources), among all feasible configurations of node electrical potentials, the one physically realized minimizes the power dissipated in the resistors. Other related results include Onsager's principle \citep{onsager1931reciprocal} and Millar's theory \citep{millar1951cxvi}, which provides a variational formulation for the steady state in a network composed of elements with arbitrary current-voltage (i-v) characteristics.

Finally, as a sum of convex functions, the energy $E(v)$ is a convex function of the node electrical potentials $v$. Specifically, the energy function $E(v)$ is a quadratic form in $v$, the Hessian of $E$ being positive definite. The feasible set $\mathcal{S}$, defined by linear constraints, is also convex. Therefore, Theorem~\ref{thm:convex-qp-formulation} states that the steady state configuration of node electrical potentials is the solution of a convex optimization problem, specifically, a convex quadratic programming (QP) problem with linear constraints.

Next we turn to our algorithm for simulating (ideal) nonlinear resistive networks.

\subsection{An Algorithm to Simulate Nonlinear Resistive Networks}

Equiped with the QP formulation (Theorem~\ref{thm:convex-qp-formulation}), we introduce a numerical method to compute the steady state $v_\star$ of an ideal nonlinear resistive network \eqref{eq:energy-minimum}. Starting from some configuration $v^{(0)} \in \mathcal{S}$, our algorithm minimizes the energy function $E(v)$ with respect to $v \in \mathcal{S}$ by building a sequence of configurations $v^{(1)}, v^{(2)}, \ldots, v^{(t)}, \ldots$ in the feasible set $\mathcal{S}$ such that $E(v^{(t+1)}) \leq E(v^{(t)})$ for every $t \geq 0$. Specifically, we use an `exact coordinate descent' algorithm. The main ingredient of our algorithm is the result below.

For simplicity and clarity, we assume here that the voltage sources form a connected component -- a tree of the network -- although our algorithm can be adapted to the more general case where the voltage sources form a forest (i.e a disjoint union of multiple trees).

\begin{restatable}[Exact coordinate descent]{thm}{coordinatedescent}
\label{thm:coordinate-descent}
Let $v = (v_1, \ldots, v_N) \in \mathcal{S}$. Let $k$ be an internal node of the network, i.e. a node that does not belong to the connected component of voltage sources. Define
\begin{equation}
\label{eq:update-equation}
p_k := \frac{\sum_{j \in \mathcal{B}_{\rm R}} g_{kj} v_j + \sum_{j \in \mathcal{B}_{\rm CS}} i_{jk}^{\rm CS}}{\sum_{j \in \mathcal{B}_{\rm R}} g_{kj}}
\end{equation}
and
\begin{gather}
\label{eq:bounds}
v_k^{\rm min} := \max_{j: (j,k) \in \mathcal{B}_{\rm D}} v_j, \qquad v_k^{\rm max} := \min_{j: (k,j) \in \mathcal{B}_{\rm D}} v_j, \\
v_k^\prime := \min \left( \max \left( v_k^{\rm min}, p_k \right), v_k^{\rm max} \right), \\
v^\prime := (v_1, \ldots, v_{k-1}, v_k^\prime, v_{k+1}, \ldots, v_N).
\end{gather}
Then, among all configurations $v'' \in \mathcal{S}$ of the form $v''=(v_1, \ldots, v_{k-1}, v_k'', v_{k+1}, \ldots, v_N)$, the configuration $v^\prime$ is the one with the lowest energy, i.e.
\begin{equation}
v_k^\prime = \underset{v_k'' \in \mathcal{S}}{\arg \min} \; E(v_1, \ldots, v_{k-1}, v_k'', v_{k+1}, \ldots, v_N).
\end{equation}
In particular $v^\prime \in \mathcal{S}$ and $E(v^\prime) \leq E(v)$.
\end{restatable}

We prove Theorem~\ref{thm:coordinate-descent} in Appendix \ref{sec:proofs}. Below, we provide an intuitive explanation of the result. Let $k$ such that $1 \leq k \leq N$. The energy function $E$ of Eq.~\eqref{eq:energy-function} is a quadratic function of $v_k$ given the state of other node electrical potentials fixed ($v_1, \ldots, v_{k-1}, v_{k+1}, \ldots, v_N$). In other words, $E$ is of the form $E = a_k v_k^2 + b_k v_k + c_k$ for some real-valued coefficients $a_k$, $b_k$ and $c_k$ that do not depend on $v_k$. Specifically, the values of $a_k$ and $b_k$ are $a_k:=\frac{1}{2}\sum_{j \in \mathcal{B}_{\rm R}} g_{jk}$ and $b_k:=-\sum_{j \in \mathcal{B}_{\rm R}} g_{kj} v_j - \sum_{j \in \mathcal{B}_{\rm CS}} i_{jk}^{\rm CS}$. Furthermore, the range of feasible values for $v_k$, constrained by the diodes, is of the form $[v_k^{\rm min}, v_k^{\rm max}]$ where $v_k^{\rm min}$ and $v_k^{\rm max}$ are given by Eq.~\eqref{eq:bounds}. Since the coefficient $a_k$ is positive (as a sum of conductances), $E(v_k)$ is bounded below and its minimum in $\mathbb{R}$ is obtained at $p_k := - b_k / 2 a_k$. The minimum of $E(v_k)$ in the interval $[v_k^{\rm min}, v_k^{\rm max}]$ is found by clipping $p_k$ between $v_k^{\rm min}$ and $v_k^{\rm max}$. These boundaries reflect the constraints imposed by the diodes connected to node $k$, introducing nonlinearities.

Using Theorem~\ref{thm:coordinate-descent}, we can minimize the energy function $E$ via an `exact coordinate descent' strategy. Assuming that the feasible set $\mathcal{S}$ is not empty, we start from some feasible configuration $v \in \mathcal{S}$. Then, at each step, we pick some internal node $k$, and compute the value $v_k^\prime$ that minimizes $E$ given the values of other variables fixed, using Theorem~\ref{thm:coordinate-descent}. This yields a new feasible configuration $v^\prime \in \mathcal{S}$. We then pick another variable and repeat the process. At each step, the energy $E$ either remains constant or decreases.

\section{Deep Resistive Networks}
\label{sec:deep-resistive-network}

We now turn to the deep resistive network (DRN) model, a layered nonlinear resistive network architecture introduced in \citet{kendall2020training}. DRNs, which take inspiration from the architecture of layered neural networks, offer advantages for hardware design as they are potentially amenable for implementation using crossbar arrays of memristors \citep{xia2019memristive}. In this section, we show that ideal DRNs also present advantages in terms of simulation speed, when our algorithm is executed on multi-processor computers such as GPUs.

\begin{figure*}[ht!]
\begin{center}
\fbox{
\includegraphics[width=0.95\textwidth]{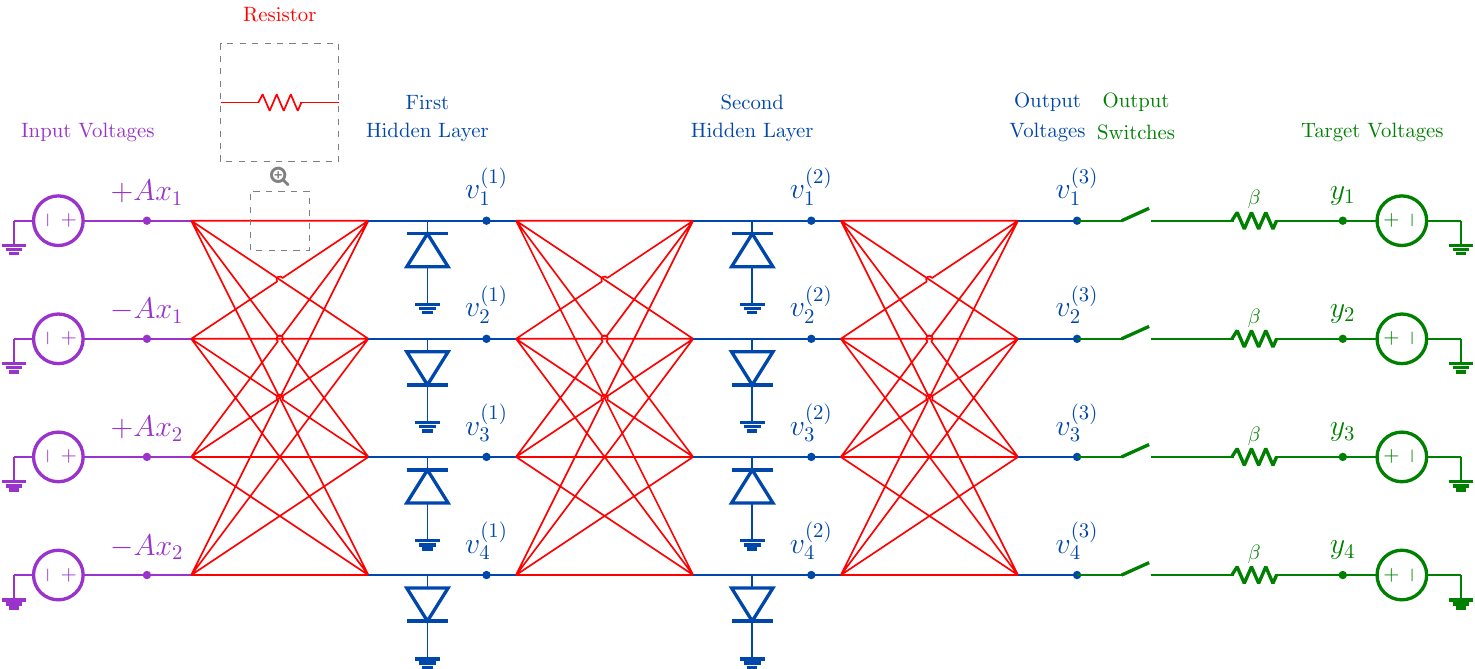}
}
\vspace{0.5cm}
\includegraphics[width=0.8\textwidth]{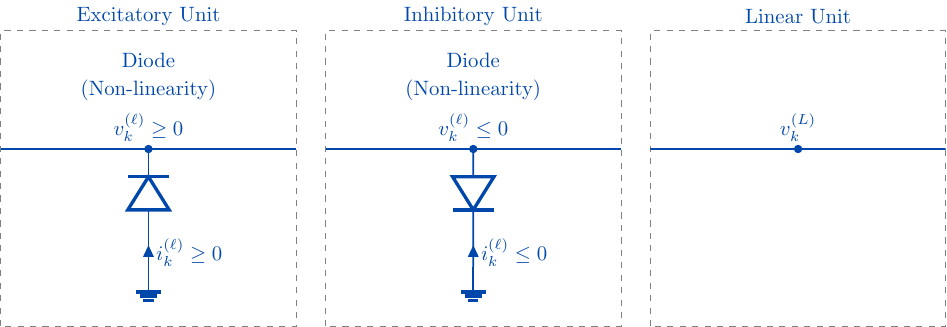}
\end{center}
\caption{
\textbf{Top.} A deep resistive network (DRN) with $L=3$ layers. Input voltage sources are set to input values: $v_1^{(0)} = A x_1$, $v_2^{(0)} = - A x_1$, $v_3^{(0)} = A x_2$ and $v_4^{(0)} = - A x_2$, where $A$ is the input amplification factor. At inference, output switches are open. Equilibrium propagation learning requires \textit{nudging} the output node voltages ($v_1^{(3)}$, $v_2^{(3)}$ $v_3^{(3)}$ and $v_4^{(3)}$) towards the target voltages ($y_1$, $y_2$, $y_3$ and $y_4$), which is achieved by closing the output switches. In the DRN architecture, the update rule for a given unit prescribed by exact coordinate descent depends only on the states of the units of the previous layer and the next layer. We can thus update the even layers ($\ell=2$) simulatneously, and then update all the odd layers ($\ell=1$ and $\ell=3$) simultaneously. This is called exact block coordinate descent. \textbf{Bottom.} To form a nonlinear unit, we place a diode between the unit's node and ground. Depending on the orientation of the diode, the units come in two flavours: excitatory units and inhibitory units.
}
\label{fig:deep-resistive-network}
\end{figure*}

\subsection{Deep Resistive Network Architecture}

The DRN architecture is defined as follows. First of all, we choose a reference node called `ground'. A set of nodes connected to ground by voltage sources form the `input layer'. These voltage sources can be set to input values, playing the role of input variables. Another set of nodes form the `output layer' whose electrical potentials play the role of model outputs.

\paragraph{Energy function.}
In a DRN, the circuit elements (voltage sources, resistors, diodes and current sources) are assembled into a layered network, mimicking the architecture of a deep neural network (Figure \ref{fig:deep-resistive-network}). For each $\ell$ such that $0 \leq \ell \leq L$, we denote $N_\ell$ the number of nodes in layer $\ell$, where $L$ is the number of layers in the DRN. Each node in the network is also called a `unit' by analogy with a neural network. We denote $v_k^{(\ell)}$ the electrical potential of the $k$-th node of layer $\ell$, which we may think of as the unit's activation. Pairs of nodes from two consecutive layers are interconnected by variable resistors - the `trainable weights'. Denoting $g_{jk}^{(\ell)}$ the conductance of the variable resistor between the $j$-th node of layer $\ell-1$ and the $k$-th node of layer $\ell$, the energy function \eqref{eq:energy-function} of a DRN takes the form
\begin{equation}
\label{eq:energy-function-drn}
E(v) = \frac{1}{2} \sum_{\ell=1}^{L} \sum_{j=1}^{N_{\ell-1}} \sum_{k=1}^{N_\ell} g_{jk}^{(\ell)} \left( v_j^{(\ell-1)} - v_k^{(\ell)} \right)^2.
\end{equation}

\paragraph{Feasible set.}
The voltage sources and resistors being linear elements, the network built thus far is linear. To make it nonlinear, we use diodes. For each unit, we place a diode between the unit's node and ground, which can be oriented in either of the two directions. If the diode points from ground to the unit's node, the unit's electrical potential is non-negative: we call it an `excitatory unit'. Conversely, if the diode points from the unit's node to ground, the unit's electrical potential is non-positive: we call it an `inhibitory unit'. For each internal (`hidden') layer $\ell$ of the DRN ($1 \leq \ell \leq L-1$), we orient the diodes so that the units of even indices are excitatory and the units of odd indices are inhibitory. Finally, the units of the output layer ($\ell=L$) are linear, i.e. they do not possess diodes. The feasible set \eqref{eq:feasible-set} corresponding to this DRN architecture is
\begin{gather}
\label{eq:feasible-set-drn}
\mathcal{S} = \{ v \in \mathbb{R}^{\sum_{\ell=1}^L N_\ell} \mid v_k^{(\ell)} \geq 0 \text{ if k is even}, \\
v_k^{(\ell)} \leq 0 \text{ if k is odd}, 1 \leq \ell \leq L-1, 1 \leq k \leq N_\ell \}.
\end{gather}

\paragraph{Input voltage sources.}
One constraint with resistive networks in general, and deep resistive networks in particular, is the non-negativity of the weights - a conductance is non-negative. By using both excitatory units and inhibitory units in the hidden layers, we have partially overcome this constraint. To further enhance the representational capacity of the DRN, we also double the number of units in the input layer -- i.e. we choose $N_0 = 2 \dim(x)$ where $\dim(x)$ is the dimension of input $x$ -- and we set the input voltage sources such that $v_{2k}^{(0)} = - v_{2k-1}^{(0)}$ for each $1 \leq k \leq N_0$. Another constraint with DRNs is the decay in amplitude of the layers' voltage values, as the depth of the network increases. We overcome this constraint by amplifying the input voltages by a factor $A \gg 1$, so that $v_{2k-1}^0 = +A x_k$ and $v_{2k}^0 = -A x_k$ for every $k$, where $x_k$ is the $k$-th input value ($1 \leq k \leq N_0$).

\paragraph{Output switches.}
Finally, additional circuitry is needed for equilibrium propagation (EP) learning. Each output node is linked to ground via a switch in series with a resistor and a voltage source. The voltage sources are used to set the desired output values, playing the role of targets. The resistors all share a common conductance value $\beta > 0$. During inference, the switches are open, whereas in the training phase of EP, the switches are closed to drive the state of output units towards the desired output values -- see Appendix~\ref{sec:equilibrium-propagation} for a brief presentation of EP \citep{scellier2017equilibrium}.

\subsection{A Fast Algorithm to Simulate Deep Resistive Networks}
\label{sec:block-coordinate-descent}

In a DRN, the update rules (Theorem~\ref{thm:coordinate-descent}) for the units take the following form. For every $(\ell,k)$ such that $1 \leq \ell \leq L-1$ and $1 \leq k \leq N_\ell$, the update rule for $v_k^{(\ell)}$ is
\begin{gather}
\label{eq:update-rule-hidden-drn}
p_k^{(\ell)} := \frac{\sum_{j=1}^{N_{\ell-1}} g_{jk}^{(\ell)} v_j^{(\ell-1)} + \sum_{j=1}^{N_{\ell+1}} g_{kj}^{(\ell+1)} v_j^{(\ell+1)}}{\sum_{j=1}^{N_{\ell-1}} g_{jk}^{(\ell)} + \sum_{j=1}^{N_{\ell+1}} g_{kj}^{(\ell+1)}} \\
v_k^{(\ell)} \leftarrow 
\left\{
\begin{array}{l}
\max \left( 0, p_k^{(\ell)} \right) \qquad \text{if k is even (excitatory unit)}, \\
\min \left( 0, p_k^{(\ell)} \right) \qquad \text{if k is odd (inhibitory unit)}.
\end{array}
\right.
\end{gather}
Assuming that the output switches are closed, the update rule for the output unit $v_k^{(L)}$ is
\begin{equation}
\label{eq:update-rule-output-drn}
v_k^{(L)} \leftarrow 
\frac{\sum_{j=1}^{N_{\ell-1}} g_{jk}^{(L)} v_j^{(L-1)} + \beta y_k}{\sum_{j=1}^{N_{\ell-1}} g_{jk}^{(L)} + \beta}.
\end{equation}
If the output switches are open, Equation~\eqref{eq:update-rule-output-drn} still holds by setting $\beta=0$.

We now derive a specialized, fast algorithm to compute the steady state of a DRN. Since the update rule for $v_k^{(\ell)}$ depends only on the state of the units in layers $\ell-1$ and $\ell+1$, and since this is true for all the units in layer $\ell$, we may update all these units simultaneously rather than sequentially. Pushing this idea further, we can partition the layers of the network in two groups: the group of layers of even index (even $\ell$) and the group of layers of odd index (odd $\ell$). Since the update rules for the odd layers depend only on the state of even layers, and vice versa, we can compute the steady state of a DRN by updating alternatively the layers of odd indices (given the state of the layers of even indices fixed) and the layers of even indices (given the state of the layers of odd indices fixed). We obtain an `exact block coordinate descent' algorithm to simulate DRNs. The property of DRNs that make it possible is the bipartite structure of their graph, where the layers of even indices constitute one set, and the layers of odd indices constitute the other set.

Importantly, equations~\eqref{eq:update-rule-hidden-drn} and ~\eqref{eq:update-rule-output-drn} can be written in matrix-vector form. Each step of our exact block coordinate descent algorithm (updating half of the layers) thus consists in performing $\sim L$ matrix-vector multiplications, $\sim L/2$ divisions and $\sim L/2$ clipping operations. This makes it an algorithm ideally suited to run on parallel computing platforms such as GPUs.

We note that it is possible to further speed up the simulations of DRNs by computing the row-wise sums of the weight matrices (used at the denominator of the update rules) only once for each inference (steady state computation), instead of computing it at each iteration of the algorithm.

\section{Simulations}
\label{sec:simulations}

We use our exact block coordinate descent algorithm to train deep resistive networks (DRNs) with equilibrium propagation (EP) \citep{scellier2017equilibrium} on the MNIST classification task. We train DRNs of one, two and three hidden layers (each comprising 1024 units), denoted DRN-1H, DRN-2H and DRN-3H, respectively. We also train another two DRNs with a single hidden layer each, compring 32784 units and 100 units, denoted DRN-XL and DRN-XS, respectively. The DRN-XS model has the same architecture as in \citet{kendall2020training}, while the DRN-XL model is 325x larger. EP is presented in Appendix~\ref{sec:equilibrium-propagation} and a full description of the DRN models and the hyperparameters used for training are provided in Appendix~\ref{sec:simulation-details}.\footnote{The code to reproduce the results is available at \url{https://github.com/rain-neuromorphics/energy-based-learning}}

As a baseline for EP, we also trained DRNs with a version of backpropagation (BP) described in Appendix~\ref{sec:simulation-details}. This version of BP is only applicable in simulations and cannot be implemented on analog hardware in any obvious way, as it requires backpropagating through the trajectory of configurations obtained by our exact block coordinate descent algorithm.

\begin{table}[ht!]
\caption{We trained five deep resistive network (DRN) architectures (XS, XL, 1H, 2H and 3H) on MNIST. The size of the network, as measured per the number of weights (in millions), is written in brackets. Alg refers to the training algorithm: equilibrium propagation (EP) or backpropagation (BP). Our exact block coordinate descent method was used to compute the DRN steady states. Test refers to the test error rates (in \%). For each experiment, we performed five runs and we reported the mean values and std values. We also reported the number of epochs of training and the wall-clock time (WCT). A full description of the models with the hyperparameters used for training are reported in Table~\ref{table:hyperparams} of Appendix~\ref{sec:simulation-details}. We also report the results of the SPICE-based simulations of \citet{kendall2020training} used as a baseline, denoted SPICE XS. 
}
\label{table:results}
\vskip 0.15in
\begin{center}
\begin{small}
\begin{sc}
\begin{tabular}{ccccc}
\toprule
DRN & Alg & Test (\%) & Epochs & WCT \\
\midrule
SPICE XS & \multirow{2}{*}{EP}
& \multirow{2}{*}{3.43} & \multirow{2}{*}{10} & \multirow{2}{*}{1 week} \\
(0.16M) & & & & \\
\midrule
\multirow{2}{*}{XS (0.16M)} & EP & 3.46 $\pm$ 0.07 & 10 & 0:30 \\
& BP & 3.30 $\pm$ 0.13 & 10 & 0:32 \\
\multirow{2}{*}{XL (51.7M)} & EP & 1.33 $\pm$ 0.02 & 100 & 10:27 \\
& BP & 1.30 $\pm$ 0.03 & 100 & 8:19 \\
\multirow{2}{*}{1h (1.6M)} & EP & 1.57 $\pm$ 0.07 & 50 & 2:36 \\
& BP & 1.54 $\pm$ 0.04 & 50 & 2:48 \\
\multirow{2}{*}{2h (2.7M)} & EP & 1.48 $\pm$ 0.05 & 50 & 4:29 \\
& BP & 1.45 $\pm$ 0.08 & 50 & 4:53 \\
\multirow{2}{*}{3h (3.7M)} & EP & 1.66 $\pm$ 0.09 & 50 & 6:57 \\
& BP & 1.50 $\pm$ 0.07 & 50 & 7:55 \\
\bottomrule
\end{tabular}
\end{sc}
\end{small}
\end{center}
\vskip -0.1in
\end{table}

Table~\ref{table:results} shows the results.

The largest network that we train (the DRN-XL model) has 327 times as many parameters as the DRN-XS architecture used in \citet{kendall2020training} (51.7M vs 0.16M). We train it for 10 times as many epochs (100 vs 10) and the total duration of training is 16 times shorter (10 hours 27 min vs 1 week). Thus, our network-size-to-epoch-duration ratio is 50000x larger. The performance obtained with the DRN-XL model is also significantly better (1.33\% vs 3.43\% test error rate).

For all the simulations, we used mini-batches of size 4, because we found that this batch size yields the best results (test error rate). However, the simulation times of the DRN-XS, -1H, -2H and -3H models can be significantly reduced by using larger batch sizes.

\section{Discussion}

As nonlinear resistive networks have recently attracted interest as physical (analog) self-learning machines, efficiently simulating these networks has become essential to advance research in this direction. Previous works either used general purpose circuit simulators such as SPICE -- which are extremely slow as they were not specifically conceived for the simulations of nonlinear resistive networks -- or resorted to linear networks, which are easier to simulate but lack a crucial feature of machine learning: nonlinearity. In this work, we have introduced a methodology specifically tailored for the simulations of nonlinear resistive networks. We have shown that the problem of determining the steady state of a nonlinear resistive network can be formulated as a convex quadratic programming (QP) problem, where the objective function is the power dissipated in the resistors, and the feasible set of node electrical potentials reflects the inequality constraints imposed by the diodes. Using this formulation, we have introduced an efficient algorithm to simulate nonlinear resistive networks with arbitrary network topologies, based on `exact coordinate descent' \citep{wright2015coordinate}. In the case of layered networks - deep resistive networks (DRNs) - we took advantage of the bipartite structure of the network to derive an exact block coordinate descent algorithm where half of the coordinates (node electrical potentials) are updated at each step. Each step of our algorithm involves solely matrix-vector multiplications, divisions and clipping, making it ideal to run on parallel processors such as GPUs. Compared to the SPICE-based simulations of \citet{kendall2020training}, the largest networks that we trained are 327 times larger, and simulating them (on a single A100 GPU) was 160 times faster. Training larger networks for more epochs also helped us achieve significantly better results on the MNIST dataset (1.33\% vs 3.43\% test error rate).

Although we have primarily focused on layered architectures (DRNs), our algorithm applies to arbitrary network topologies, including unstructured (disordered) networks \citep{stern2022physical,stern2024training,wycoff2022desynchronous}. In such networks, parallelization is also possible, although it is less straightforward to implement. More generally, our characterization of the steady state as the solution of a quadratic programming (QP) problem with linear constraints offers other options to simulate nonlinear resistive networks. Indeed, QP problems can be solved with various algorithms, such as primal-dual interior point methods (IPM) and sequential quadratic programming (SQP) methods. Many optimization libraries and software packages provide specialized QP solvers. Besides, while in our simulations we have used equilibrium propagation (EP) for training, our exact coordinate descent algorithm for nonlinear resistive networks can be used in conjunction with other learning algorithms such as `frequency propagation' \citep{anisetti2024frequency} and `agnostic EP' \citep{scellier2022agnostic}. It can also be used in conjunction with e.g. the method proposed by \citet{stern2024training} to mitigate power dissipation.

While our methodology to simulate nonlinear resistive networks is fairly general, it is also important to acknowledge the limitations of our approach. Our methodology is anchored in an assumption of ideality of the circuit elements (resistors, diodes, voltage sources and current sources). Real-world diodes, however, deviate significantly from the ideal model studied here: they have a forward voltage drop, reverse leakage current, non-zero resistance when forward-biased, and a breakdown voltage when reverse-biased. Despite these assumptions, which may seem overly simplistic, we believe that our algorithm can be usefully applied to prototype and explore a wide range of network topologies. Our methodology can also be extended to the simulations of electrical networks composed of elements possessing more realistic i-v curves (e.g. piece-wise linear i-v curves), but we leave this for future work to investigate.

Looking forward, our simulation methodology can foster more rapid progress in nonlinear resistive network research. It offers the perspective to perform large scale simulations of deep resistive networks (DRNs), to enable further assessment of the scalability of such physical (analog) self-learning machines on more complex tasks. In this respect, a related line of works on continuous Hopfield networks holds promise. Nonlinear resistive networks and DRNs are, indeed, closely related to continuous Hopfield networks \citep{hopfield1984neurons} and their layered version, the `deep Hopfield network' (DHN) \citep{scellier2017equilibrium}. In particular, our exact coordinate descent algorithm for resistive networks is similar in spirit to the asynchronous update scheme for Hopfield networks, and our exact block coordinate descent algorithm for DRNs is similar to the asynchronous energy minimization scheme for DHNs \citep{scellier2023energy} and to block-Gibbs sampling in deep Boltzmann machines \citep{salakhutdinov2009deep}. Importantly, unlike Hopfield networks and Boltzmann machines whose energy functions are typically non-convex, the energy function (power dissipation) of a nonlinear resistive network is convex. In practice we find that our DRN simulations require much fewer steps to converge than in DHNs. What is especially interesting is that recent works have shown that DHNs yield promising results on image classification tasks such as CIFAR-10 \citep{laborieux2021scaling}, CIFAR-100 \citep{scellier2023energy} and ImageNet 32x32 \citep{laborieux2022holomorphic}. A more detailed analysis of the similarities and differences between resistive networks and continuous Hopfield networks is provided in Appendix~\ref{sec:continuous-hopfield-networks}.

\section*{Acknowledgements}

The author thanks the anonymous reviewers of ICML for their useful comments, as well as Sid Mishra, Jack Kendall, Maxence Ernoult, Mohammed Fouda, Suhas Kumar and Jeremie Laydevant for useful feedback and discussions.

\section*{Impact Statement}

This work presents an algorithm that could enable large-scale modelling of energy-efficient hardware for machine learning (ML), and therefore facilitate the development of such hardware. Such energy-efficient hardware would significantly reduce the costs associated with inference and training of ML models, and therefore make advanced AI technologies more affordable and accessible to a wider range of populations and industrial sectors. Such hardware would also open the door to train and deploy significantly larger AI models.

\bibliographystyle{icml2024}
\bibliography{biblio}

\begin{thebibliography}{39}
\providecommand{\natexlab}[1]{#1}
\providecommand{\url}[1]{\texttt{#1}}
\expandafter\ifx\csname urlstyle\endcsname\relax
  \providecommand{\doi}[1]{doi: #1}\else
  \providecommand{\doi}{doi: \begingroup \urlstyle{rm}\Url}\fi

\bibitem[Altman et~al.(2023)Altman, Stern, Liu, and Durian]{altman2023experimental}
Altman, L.~E., Stern, M., Liu, A.~J., and Durian, D.~J.
\newblock Experimental demonstration of coupled learning in elastic networks.
\newblock \emph{arXiv preprint arXiv:2311.00170}, 2023.

\bibitem[Anisetti et~al.(2024)Anisetti, Kandala, Scellier, and Schwarz]{anisetti2024frequency}
Anisetti, V.~R., Kandala, A., Scellier, B., and Schwarz, J.
\newblock Frequency propagation: Multimechanism learning in nonlinear physical networks.
\newblock \emph{Neural Computation}, pp.\  1--25, 2024.

\bibitem[Dillavou et~al.(2022)Dillavou, Stern, Liu, and Durian]{dillavou2022demonstration}
Dillavou, S., Stern, M., Liu, A.~J., and Durian, D.~J.
\newblock Demonstration of decentralized physics-driven learning.
\newblock \emph{Physical Review Applied}, 18\penalty0 (1):\penalty0 014040, 2022.

\bibitem[Dillavou et~al.(2023)Dillavou, Beyer, Stern, Miskin, Liu, and Durian]{dillavou2023machine}
Dillavou, S., Beyer, B.~D., Stern, M., Miskin, M.~Z., Liu, A.~J., and Durian, D.~J.
\newblock Machine learning without a processor: Emergent learning in a nonlinear electronic metamaterial.
\newblock \emph{arXiv preprint arXiv:2311.00537}, 2023.

\bibitem[Ernoult et~al.(2019)Ernoult, Grollier, Querlioz, Bengio, and Scellier]{ernoult2019updates}
Ernoult, M., Grollier, J., Querlioz, D., Bengio, Y., and Scellier, B.
\newblock Updates of equilibrium prop match gradients of backprop through time in an rnn with static input.
\newblock \emph{Advances in neural information processing systems}, 32, 2019.

\bibitem[Falk et~al.(2023)Falk, Strupp, Scellier, and Murugan]{falk2023contrastive}
Falk, M., Strupp, A., Scellier, B., and Murugan, A.
\newblock Contrastive learning through non-equilibrium memory.
\newblock \emph{arXiv preprint arXiv:2312.17723}, 2023.

\bibitem[Hopfield(1984)]{hopfield1984neurons}
Hopfield, J.~J.
\newblock Neurons with graded response have collective computational properties like those of two-state neurons.
\newblock \emph{Proceedings of the national academy of sciences}, 81\penalty0 (10):\penalty0 3088--3092, 1984.

\bibitem[Keiter(2014)]{keiter2014xyce}
Keiter, E.
\newblock Xyce: An open source spice engine, Mar 2014.
\newblock URL \url{https://nanohub.org/resources/20605}.

\bibitem[Kendall et~al.(2020)Kendall, Pantone, Manickavasagam, Bengio, and Scellier]{kendall2020training}
Kendall, J., Pantone, R., Manickavasagam, K., Bengio, Y., and Scellier, B.
\newblock Training end-to-end analog neural networks with equilibrium propagation.
\newblock \emph{arXiv preprint arXiv:2006.01981}, 2020.

\bibitem[Kiraz et~al.(2022)Kiraz, Pham, and Desgreys]{kiraz2022impacts}
Kiraz, F.~Z., Pham, D.-K.~G., and Desgreys, P.
\newblock Impacts of feedback current value and learning rate on equilibrium propagation performance.
\newblock In \emph{2022 20th IEEE Interregional NEWCAS Conference (NEWCAS)}, pp.\  519--523. IEEE, 2022.

\bibitem[Laborieux \& Zenke(2022)Laborieux and Zenke]{laborieux2022holomorphic}
Laborieux, A. and Zenke, F.
\newblock Holomorphic equilibrium propagation computes exact gradients through finite size oscillations.
\newblock \emph{Advances in Neural Information Processing Systems}, 35:\penalty0 12950--12963, 2022.

\bibitem[Laborieux et~al.(2021)Laborieux, Ernoult, Scellier, Bengio, Grollier, and Querlioz]{laborieux2021scaling}
Laborieux, A., Ernoult, M., Scellier, B., Bengio, Y., Grollier, J., and Querlioz, D.
\newblock Scaling equilibrium propagation to deep convnets by drastically reducing its gradient estimator bias.
\newblock \emph{Frontiers in neuroscience}, 15:\penalty0 129, 2021.

\bibitem[Laydevant et~al.(2024)Laydevant, Markovi{\'c}, and Grollier]{laydevant2024training}
Laydevant, J., Markovi{\'c}, D., and Grollier, J.
\newblock Training an ising machine with equilibrium propagation.
\newblock \emph{Nature Communications}, 15\penalty0 (1):\penalty0 3671, 2024.

\bibitem[LeCun et~al.(1998)LeCun, Bottou, Bengio, and Haffner]{lecun1998gradient}
LeCun, Y., Bottou, L., Bengio, Y., and Haffner, P.
\newblock Gradient-based learning applied to document recognition.
\newblock \emph{Proceedings of the IEEE}, 86\penalty0 (11):\penalty0 2278--2324, 1998.

\bibitem[Markovi{\'c} et~al.(2020)Markovi{\'c}, Mizrahi, Querlioz, and Grollier]{markovic2020physics}
Markovi{\'c}, D., Mizrahi, A., Querlioz, D., and Grollier, J.
\newblock Physics for neuromorphic computing.
\newblock \emph{Nature Reviews Physics}, 2\penalty0 (9):\penalty0 499--510, 2020.

\bibitem[Martin et~al.(2021)Martin, Ernoult, Laydevant, Li, Querlioz, Petrisor, and Grollier]{martin2021eqspike}
Martin, E., Ernoult, M., Laydevant, J., Li, S., Querlioz, D., Petrisor, T., and Grollier, J.
\newblock Eqspike: spike-driven equilibrium propagation for neuromorphic implementations.
\newblock \emph{Iscience}, 24\penalty0 (3), 2021.

\bibitem[Massar \& Mognetti(2024)Massar and Mognetti]{massar2024equilibrium}
Massar, S. and Mognetti, B.~M.
\newblock Equilibrium propagation: the quantum and the thermal cases.
\newblock \emph{arXiv preprint arXiv:2405.08467}, 2024.

\bibitem[Millar(1951)]{millar1951cxvi}
Millar, W.
\newblock Cxvi. some general theorems for non-linear systems possessing resistance.
\newblock \emph{The London, Edinburgh, and Dublin Philosophical Magazine and Journal of Science}, 42\penalty0 (333):\penalty0 1150--1160, 1951.

\bibitem[Oh et~al.(2023)Oh, An, Cho, Yoon, and Min]{oh2023memristor}
Oh, S., An, J., Cho, S., Yoon, R., and Min, K.-S.
\newblock Memristor crossbar circuits implementing equilibrium propagation for on-device learning.
\newblock \emph{Micromachines}, 14\penalty0 (7):\penalty0 1367, 2023.

\bibitem[Onsager(1931)]{onsager1931reciprocal}
Onsager, L.
\newblock Reciprocal relations in irreversible processes. ii.
\newblock \emph{Physical review}, 38\penalty0 (12):\penalty0 2265, 1931.

\bibitem[Paszke et~al.(2017)Paszke, Gross, Chintala, Chanan, Yang, DeVito, Lin, Desmaison, Antiga, and Lerer]{paszke2017automatic}
Paszke, A., Gross, S., Chintala, S., Chanan, G., Yang, E., DeVito, Z., Lin, Z., Desmaison, A., Antiga, L., and Lerer, A.
\newblock Automatic differentiation in pytorch.
\newblock 2017.

\bibitem[Salakhutdinov \& Hinton(2009)Salakhutdinov and Hinton]{salakhutdinov2009deep}
Salakhutdinov, R. and Hinton, G.
\newblock Deep boltzmann machines.
\newblock In \emph{Artificial intelligence and statistics}, pp.\  448--455. PMLR, 2009.

\bibitem[Scellier(2021)]{scellier2021deep}
Scellier, B.
\newblock \emph{A deep learning theory for neural networks grounded in physics}.
\newblock PhD thesis, Universit\'e de Montr\'eal, 2021.

\bibitem[Scellier(2024)]{scellier2024quantum}
Scellier, B.
\newblock Quantum equilibrium propagation: Gradient-descent training of quantum systems.
\newblock \emph{arXiv preprint arXiv:2406.00879}, 2024.

\bibitem[Scellier \& Bengio(2017)Scellier and Bengio]{scellier2017equilibrium}
Scellier, B. and Bengio, Y.
\newblock Equilibrium propagation: Bridging the gap between energy-based models and backpropagation.
\newblock \emph{Frontiers in computational neuroscience}, 11:\penalty0 24, 2017.

\bibitem[Scellier \& Mishra(2023)Scellier and Mishra]{scellier2023universal}
Scellier, B. and Mishra, S.
\newblock A universal approximation theorem for nonlinear resistive networks.
\newblock \emph{arXiv preprint arXiv:2312.15063}, 2023.

\bibitem[Scellier et~al.(2022)Scellier, Mishra, Bengio, and Ollivier]{scellier2022agnostic}
Scellier, B., Mishra, S., Bengio, Y., and Ollivier, Y.
\newblock Agnostic physics-driven deep learning.
\newblock \emph{arXiv preprint arXiv:2205.15021}, 2022.

\bibitem[Scellier et~al.(2024)Scellier, Ernoult, Kendall, and Kumar]{scellier2023energy}
Scellier, B., Ernoult, M., Kendall, J., and Kumar, S.
\newblock Energy-based learning algorithms for analog computing: a comparative study.
\newblock \emph{Advances in Neural Information Processing Systems}, 36, 2024.

\bibitem[Stern et~al.(2022)Stern, Dillavou, Miskin, Durian, and Liu]{stern2022physical}
Stern, M., Dillavou, S., Miskin, M.~Z., Durian, D.~J., and Liu, A.~J.
\newblock Physical learning beyond the quasistatic limit.
\newblock \emph{Physical Review Research}, 4\penalty0 (2):\penalty0 L022037, 2022.

\bibitem[Stern et~al.(2024)Stern, Dillavou, Jayaraman, Durian, and Liu]{stern2024training}
Stern, M., Dillavou, S., Jayaraman, D., Durian, D.~J., and Liu, A.~J.
\newblock Training self-learning circuits for power-efficient solutions.
\newblock \emph{APL Machine Learning}, 2\penalty0 (1), 2024.

\bibitem[Vogt et~al.(2020)Vogt, Hendrix, and Nenzi]{vogt2020ngspice}
Vogt, H., Hendrix, M., and Nenzi, P.
\newblock Ngspice (version 31), 2020.
\newblock URL \url{http://ngspice.sourceforge.net/docs/ngspice-manual.pdf}.

\bibitem[Wang et~al.(2024)Wang, Wanjura, and Marquardt]{wang2024training}
Wang, Q., Wanjura, C.~C., and Marquardt, F.
\newblock Training coupled phase oscillators as a neuromorphic platform using equilibrium propagation.
\newblock \emph{arXiv preprint arXiv:2402.08579}, 2024.

\bibitem[Watfa et~al.(2023)Watfa, Garcia-Ortiz, and Sassatelli]{watfa2023energy}
Watfa, M., Garcia-Ortiz, A., and Sassatelli, G.
\newblock Energy-based analog neural network framework.
\newblock \emph{Frontiers in Computational Neuroscience}, 17:\penalty0 1114651, 2023.

\bibitem[Williams et~al.(2023)Williams, Bredenberg, and Lajoie]{williams2023flexible}
Williams, E., Bredenberg, C., and Lajoie, G.
\newblock Flexible phase dynamics for bio-plausible contrastive learning.
\newblock In \emph{International Conference on Machine Learning}, pp.\  37042--37065. PMLR, 2023.

\bibitem[Wright(2015)]{wright2015coordinate}
Wright, S.~J.
\newblock Coordinate descent algorithms.
\newblock \emph{Mathematical programming}, 151\penalty0 (1):\penalty0 3--34, 2015.

\bibitem[Wycoff et~al.(2022)Wycoff, Dillavou, Stern, Liu, and Durian]{wycoff2022desynchronous}
Wycoff, J.~F., Dillavou, S., Stern, M., Liu, A.~J., and Durian, D.~J.
\newblock Desynchronous learning in a physics-driven learning network.
\newblock \emph{The Journal of Chemical Physics}, 156\penalty0 (14), 2022.

\bibitem[Xia \& Yang(2019)Xia and Yang]{xia2019memristive}
Xia, Q. and Yang, J.~J.
\newblock Memristive crossbar arrays for brain-inspired computing.
\newblock \emph{Nature materials}, 18\penalty0 (4):\penalty0 309--323, 2019.

\bibitem[Yi et~al.(2023)Yi, Kendall, Williams, and Kumar]{yi2023activity}
Yi, S.-i., Kendall, J.~D., Williams, R.~S., and Kumar, S.
\newblock Activity-difference training of deep neural networks using memristor crossbars.
\newblock \emph{Nature Electronics}, 6\penalty0 (1):\penalty0 45--51, 2023.

\bibitem[Zucchet \& Sacramento(2022)Zucchet and Sacramento]{zucchet2022beyond}
Zucchet, N. and Sacramento, J.
\newblock Beyond backpropagation: bilevel optimization through implicit differentiation and equilibrium propagation.
\newblock \emph{Neural Computation}, 34\penalty0 (12):\penalty0 2309--2346, 2022.

\end{thebibliography}

\clearpage
\appendix

\onecolumn

\newpage
\section{Proofs of Theorem \ref{thm:convex-qp-formulation} and Theorem \ref{thm:coordinate-descent}}
\label{sec:proofs}

In this appendix, we prove Theorem \ref{thm:convex-qp-formulation} and Theorem \ref{thm:coordinate-descent}.

First we recall our assumptions about the behaviour of individual devices. We consider four types of elements: linear resistors, diodes, voltage sources and current sources, with the following current-voltage (i-v) characteristics:
\begin{itemize}
\item A linear resistor follows Ohm's law, i.e. $i = g v$ where $g$ is the conductance of the resistor ($g=1/r$ where $r$ is the resistance).
\item A diode satisfies $i=0$ for $v \leq 0$, and $v=0$ for $i>0$.
\item A voltage source satisfies $v=v_0$ for some constant $v_0$, regardless of $i$.
\item A current source satisfies $i = i_0$ for some constant $i_0$, regardless of $v$. 
\end{itemize}

\subsection{Proof of Theorem \ref{thm:convex-qp-formulation}}

We denote $v^{\rm VS}$ as the set of voltages across voltage sources, and $i^{\rm CS}$ as the set of currents across current sources. For clarity, we restate Theorem \ref{thm:convex-qp-formulation}.

\convexqpformulation*

\begin{proof}[Proof of Theorem \ref{thm:convex-qp-formulation}]
Suppose there exists a configuration of branch voltages and branch currents that satisfies all the current-voltage equations in every branch, as well as Kirchhoff's current law (KCL) and Kirchhoff's voltage law (KVL). We call this configuration the `steady state' and denote $v_\star=(v_1,\ldots, v_N)$ the corresponding configuration of node electrical potentials, where $N$ is the number of nodes in the network. For every branch $(j,k) \in \mathcal{B}_{\rm R}$ (resp. $\mathcal{B}_{\rm D}$, $\mathcal{B}_{\rm VS}$), we denote $i_{jk}^{\rm R}$ (resp. $i_{jk}^{\rm D}$, $i_{jk}^{\rm VS}$) the current through it. The resistor's equation (Ohm's law) imposes that
\begin{equation}
\forall (j,k)\in \mathcal{B}_{\rm R}, \qquad i_{jk}^{\rm R} = g_{jk} (v_j - v_k).
\end{equation}
Let us define the functions $G_{jk}$ (for $(j,k)\in \mathcal{B}_{\rm D}$) and $H_{jk}$ (for $(j,k)\in \mathcal{B}_{\rm VS}$)  as
\begin{equation}
\label{eq:GH}
\forall (j,k)\in \mathcal{B}_{\rm D}, \quad G_{jk}(v) := v_j - v_k \qquad \text{and} \qquad \forall (j,k)\in \mathcal{B}_{\rm VS}, \quad H_{jk}(v) := v_j - v_k - v_{jk}^{\rm VS}.
\end{equation}
The diode equations impose that
\begin{equation}
\label{eq:D}
\forall (j,k)\in \mathcal{B}_{\rm D}, \qquad G_{jk}(v) \leq 0, \qquad i_{jk}^{\rm D} \geq 0, \qquad i_{jk}^{\rm D} \cdot G_{jk}(v) = 0.
\end{equation}
The voltage source equality constraints impose that
\begin{equation}
\label{eq:VS}
\forall (j,k)\in \mathcal{B}_{\rm VS}, \qquad H_{jk}(v) = 0.
\end{equation}
Next, let $k$ be a node such that $1 \leq k \leq N$, and consider all the branches connected to node $k$. These branches may include resistors, current sources, diodes and voltage sources. KCL applied to node $k$ reads
\begin{equation}
\label{eq:KCL}
\sum_{j:(j,k)\in \mathcal{B}_{\rm R}}  i_{jk}^{\rm R} + \sum_{j:(j,k)\in \mathcal{B}_{\rm CS}} i_{jk}^{\rm CS} + 
\sum_{j:(j,k)\in \mathcal{B}_{\rm D}} i_{jk}^{\rm D} +
\sum_{j:(j,k)\in \mathcal{B}_{\rm VS}} i_{jk}^{\rm VS} = 0.
\end{equation}
Let us introduce the following function $L$, which is a function of the node electrical potentials ($v$), the currents through the diodes ($i^{\rm D}$), and the currents through the voltage sources ($i^{\rm VS}$),
\begin{equation}
L(v, i^{\rm D}, i^{\rm VS}) := E(v) + \sum_{(j,k)\in \mathcal{B}_{\rm D}} i_{jk}^{\rm D} G_{jk}(v) + \sum_{(j,k)\in \mathcal{B}_{\rm VS}} i_{jk}^{\rm VS} H_{jk}(v).
\end{equation}
Using the expressions of the energy function \eqref{eq:energy-function} and the functions $G_{jk}$ and $H_{jk}$ from \eqref{eq:GH}, KCL \eqref{eq:KCL} can be rewritten in terms of $L$ as 
\begin{equation}
\label{eq:stationarity}
\nabla_v L(v,i_{jk}^{\rm D},i_{jk}^{\rm VS}) = 0.
\end{equation}
Equations \eqref{eq:D}, \eqref{eq:VS} and \eqref{eq:stationarity} fully characterize the steady state of the nonlinear resistive network. These equations constitute the set of Karush-Kuhn-Tucker (KKT) conditions associated with the constrained optimization problem of Eq.~\eqref{eq:energy-minimum}. $L$ is the Lagrangian, and the currents $i^{\rm D}$ and $i^{\rm VS}$ are the Lagrange multipliers associated with the inequality and equality constraints, respectively. Equations \eqref{eq:D} and \eqref{eq:VS} constitute the primal feasibility condition (ensuring $v$ belongs to the feasible set $\mathcal{S}$), the dual feasibility condition (ensuring the Lagrange multipliers associated with the inequality constraints are non-negative), and the complementary slackness for inequality constraints (either the constraint is `binding' or the associated Lagrange multiplier is zero). Finally, \eqref{eq:stationarity} is the stationarity condition of the Lagrangian.

From the above analysis, it follows that the steady state of the network (if it exists) satisfies the KKT conditions associated with the constrained optimization problem~\eqref{eq:energy-minimum}. Conversely, a configuration of node voltages and branch currents (diode branches and voltage source branches) that satisfies the KKT conditions is a steady state.

To conclude, we note that, as a sum of convex functions, the energy function $E$ is convex. Therefore the KKT conditions are equivalent to global optimality. This implies that a steady state exists if and only if the feasible set is non-empty, in which case the corresponding configuration of node voltages is a global minimum of the convex optimization problem.
\end{proof}

\subsection{Proof of Theorem \ref{thm:coordinate-descent}}

The optimization problem is a quadratic programming (QP) problem with linear constraints. We solve it using exact coordinate descent.

\coordinatedescent*

\begin{proof}[Proof of Theorem \ref{thm:coordinate-descent}]
The energy function, as a function of $v_k$, is a second-order polynomial of the form:
\begin{align}
E(v_k) & = \frac{1}{2} \sum_{j : (j,k) \in \mathcal{B}_{\rm R}} g_{jk} \left( v_j - v_k \right)^2 + \sum_{j : (j,k) \in \mathcal{B}_{\rm CS}} i_{jk}^{\rm CS} (v_j-v_k) \\
& = \underbrace{\left( \frac{1}{2} \sum_{j : (j,k) \in \mathcal{B}_{\rm R}} g_{jk} \right)}_{=: a_k} v_k^2 - \underbrace{\left( \sum_{j : (j,k) \in \mathcal{B}_{\rm R}} g_{kj} v_j + \sum_{j : (j,k) \in \mathcal{B}_{\rm CS}} i_{jk}^{\rm CS} \right)}_{=: b_k} v_k + \text{constant}
\end{align}
The minimum (in $\mathbb{R}$) of this second-order polynomial is achieved at the point $v_k = p_k$, defined as
\begin{equation}
p_k := - \frac{b_k}{2 a_k} = \frac{\sum_{j : (j,k) \in \mathcal{B}_{\rm R}} g_{kj} v_j + \sum_{j : (j,k) \in \mathcal{B}_{\rm CS}} i_{jk}^{\rm CS}}{\sum_{j : (j,k) \in \mathcal{B}_{\rm R}} g_{jk}}.
\end{equation}
However, the range of feasible values for $v_k$ (for which the corresponding configurations are in the feasible set $\mathcal{S}$) is not $\mathbb{R}$, but $(v_k^{\rm min},v_k^{\rm max})$, where
\begin{equation}
v_k^{\rm min} := \max_{j: (j,k) \in \mathcal{B}_{\rm D}} v_j, \qquad v_k^{\rm max} := \min_{j: (k,j) \in \mathcal{B}_{\rm D}} v_j.
\end{equation}
It is easily checked that the value $v_k^\prime \in (v_k^{\rm min},v_k^{\rm max})$ that achieves the minimum of the above second-order polynomial (given other variables fixed) is
\begin{equation}
v_k^\prime := \min \left( \max \left( v_k^{\rm min}, p_k \right), v_k^{\rm max} \right).
\end{equation}
\end{proof}

\clearpage
\section{Equilibrium Propagation}
\label{sec:equilibrium-propagation}

In this appendix, we briefly present the equilibrium propagation (EP) algorithm used in section~\ref{sec:simulations} to train the deep resistive network (DRN). EP is a method for extracting the weight gradients of arbitrary cost functions in arbitrary energy-based systems, that is, systems that possess an energy function and seek an energy minimum \citep{scellier2017equilibrium} (see also \citet{scellier2021deep} for a physics-oriented presentation of EP). Besides the nonlinear resistive networks considered in this work, EP has been used in continuous Hopfield networks (see Appendix~\ref{sec:continuous-hopfield-networks}), Ising machines \citep{laydevant2024training}, coupled phase oscillators \citep{wang2024training} and spiking networks \citep{martin2021eqspike}. A variant of EP has also been used in elastic networks \citep{altman2023experimental}. Most recently, EP has been extended to quantum systems \citep{massar2024equilibrium,scellier2024quantum}.

\subsection{EP in Nonlinear Resistive Networks}
\label{sec:ep-nonlinear-resistive-network}

Consider a nonlinear resistive network with arbitrary topology, as described in Section~\ref{sec:nonlinear-resistive-network} (see Figure ~\ref{fig:augmented-circuit}, top). In this electrical network, a set of voltage sources serves as `input variables' with voltages set to the input values, while other branch voltages and currents settle to a steady state, i.e. a minimum of the energy function \eqref{eq:energy-function}, characterized by
\begin{equation}
v_\star := \underset{v \in \mathcal{S}}{\arg \min} \; E(v).
\end{equation}
A subset of these branches functions as `output variables', with their voltages used as the network's prediction.\footnote{For example, in the DRN architecture of Figure~\ref{fig:deep-resistive-network}, input voltage sources link input nodes to ground, and output branches link output nodes to ground.} We view this input-output function implemented by the network as being parameterized by the conductances of the resistors, while the diodes introduce nonlinearities. The goal is to adjust the conductance values of the resistors (the `trainable weights') so the network implements a desired input-output function. Several learning algorithms for such resistive networks have been proposed \citep{kendall2020training,dillavou2022demonstration,anisetti2024frequency} based on EP. In this work, we used EP for its simplicity and superior performance compared to other contrastive learning methods \citep{scellier2023energy}.

EP involves augmenting the network by adding a branch between each pair of output nodes. Each branch consists of a voltage source, a resistor and a switch in series (Figure ~\ref{fig:augmented-circuit}, bottom left). All output resistors share a common conductance value, $\beta > 0$. The switches have two states: they are either all open (the `free state') or all closed (the `nudge state'). In the free state, no current flows through these output branches, leaving the network's state unchanged. In the nudge state, however, currents flow through these output branches, influencing the network's state. Mathematically, in the nudge state, the network's energy function $E$ (defined in Theorem~\ref{thm:convex-qp-formulation}) is augmented by a term equal to the power dissipated in the output resistors, $\frac{1}{2} \sum_{(j,k) \in \mathcal{B}_{\rm O}} \beta \left( v_{jk} - y_{jk}  \right)^2$, with $\mathcal{B}_{O}$ denoting the set of output branches, and $y_{jk}$ the voltage across the voltage source of output branch $(j,k)$. Thus, in the nudge state, the total energy function (total power dissipation) of the augmented network is
\begin{equation}
F(\beta,v) := E(v) + \beta C(v),
\end{equation}
where
\begin{equation}
C(v) := \frac{1}{2} \sum_{(j,k) \in \mathcal{B}_{\rm O}} \left( v_{jk} - y_{jk}  \right)^2.
\end{equation}
Importantly, the function $C$ appearing in the `total energy function' is also the cost function that the network aims to minimize from a learning perspective, representing the squared error between the model prediction and the desired output. In the free state (when the output switches are all open), the energy function of the network is $F(0,v) = E(v)$.

EP operates as follows:
\begin{enumerate}
\item Pick an input-output pair $(x,y)$ from the training dataset. 
\item Set input voltage sources to the input values ($x$) and open all the output switches, allowing the network to reach the free state:
\begin{equation}
v_\star^0 := \underset{v \in \mathcal{S}}{\arg \min} \; F(0,v) = v_\star.
\end{equation}
\item Set output voltage sources to the desired output values ($y$) and close all the output switches, allowing the network to reach the nudge state:
\begin{equation}
v_\star^\beta := \underset{v \in \mathcal{S}}{\arg \min} \;  F(\beta,v).
\end{equation}
\item For every branch $(j,k)$ containing a variable resistor, update the conductance $g_{jk}$ according to the contrastive rule:
\begin{align}
\label{eq:contrastive-rule}
\Delta g_{jk} & = \frac{\eta}{\beta} \left[ \frac{\partial F}{\partial g_{jk}}(0,v_\star^0) - \frac{\partial F}{\partial g_{jk}}(\beta,v_\star^\beta) \right] \\
& = \frac{\eta}{2\beta} \left[ \left( v_{jk}^0 \right)^2 -  \left( v_{jk}^\beta \right)^2 \right],
\label{eq:contrastive-rule-resistive}
\end{align}
where $\eta > 0$ is the learning rate, $v_{jk}^0$ denotes the voltage across branch $(i,j)$ in the free state, and $v_{jk}^\beta$ denotes the voltage in the nudge state. We have used the explicit form of the energy function, as in Eq.~\eqref{eq:energy-function}.
\end{enumerate}

\begin{figure*}[h!]
\begin{center}
\includegraphics[width=\textwidth]{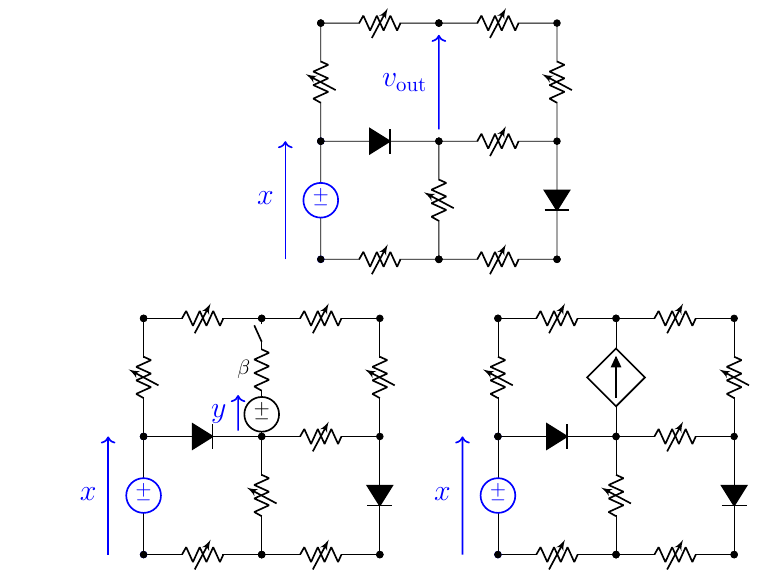}
\end{center}
\caption{
Training a nonlinear resistive network with equilibrium propagation (EP). \textbf{Top}. Input voltage $x$ is supplied to the network, and the output voltage $v_{\rm out}$ is measured. Diodes implement nonlinearities and variable resistors implement the `trainable weights'. \textbf{Bottom}. Two methods to implement `nudging' for EP learning. In the first method (\textbf{left}), a voltage source set to desired output $y$, a resistor of conductance $\beta$ and a switch are in series in the output branch. Closing the switch injects a current proportional to the prediction error, $i_{\rm out} = \beta (y-v_{\rm out})$. One caveat of this method is that the nudging parameter $\beta$ (the conductance) is necessarily positive. In the second method (\textbf{right}), a current source injects a current $i_{\rm out} = \beta (y-v_{\rm out})$ in the output branch, allowing for the use of a negative $\beta$.
\label{fig:augmented-circuit}
}
\end{figure*}

\subsection{Equilibrium Propagation Formulas}
\label{sec:ep-formulas}

The central result of EP is that the learning rule \eqref{eq:contrastive-rule} approximates one step of gradient descent on the cost function $C$ \citep{scellier2017equilibrium}. Specifically,
\begin{equation}
\Delta g_{jk} = - \eta \frac{\partial C(v_\star^0)}{\partial g_{jk}} + O(\beta)
\end{equation}
when $\beta \to 0$. More specifically, the learning rule \eqref{eq:contrastive-rule} performs one step of gradient descent on a surrogate function $\mathcal{L}_\beta$ that approximates the true cost function $C$ \citep{scellier2023energy}, i.e.
\begin{equation}
\Delta g_{jk} = - \eta \frac{\partial \mathcal{L}_\beta}{\partial g_{jk}} \qquad \text{where} \qquad \mathcal{L}_\beta = C(v_\star^0) + O(\beta) \quad \text{when} \quad \beta \to 0.
\end{equation}
The surrogate function $\mathcal{L}_\beta$ is the contrastive function defined as
\begin{equation}
\mathcal{L}_\beta := \frac{G(\beta)-G(0)}{\beta}, \qquad \text{where} \qquad G(\beta) := F(\beta,v_\star^\beta) = \min_{v \in \mathcal{S}} F(\beta,v).
\end{equation}
EP has been found to work better in practice when employing a negative $\beta$ in the nudge state, rather than a positive one \citep{scellier2023energy}. This can be explained as the contrastive function of EP is an upper bound of the true cost function when $\beta$ is negative, and a lower bound when it is positive, i.e.
\begin{equation}
\mathcal{L}_\beta \leq C(v_\star^0) \leq \mathcal{L}_{-\beta}, \qquad \forall \beta > 0.
\end{equation}
A centered version of EP, introduced earlier in \citet{laborieux2021scaling}, combines a negatively-perturbed state $v_\star^{-\beta}$ with a positively-perturbed one $v_\star^{+\beta}$, proving more effective. This can be explained as the contrastive function $\mathcal{L}_\beta^{\rm centered}$ of this centered variant approximates the true cost function to second order in $\beta$, i.e.
\begin{equation}
\mathcal{L}_\beta^{\rm centered} := \frac{G(\beta)-G(-\beta)}{2\beta} = C(v_\star^0) + O(\beta^2).
\end{equation}
This is the version of EP that we used in our simulations (Section~\ref{sec:simulations}).

Importantly, all these results also hold in the non-ideal regime where the network elements (diodes, voltage sources and current sources) have arbitrary i-v characteristics (not necessarily the i-v characteristics of Figure~\ref{fig:resistive-elements}). The only assumption, required to derive Eq.~\eqref{eq:contrastive-rule-resistive}, is that the resistors must follow Ohm's law (i.e. have the linear i-v characteristic depicted in Figure~\ref{fig:resistive-elements}). We refer to \citet{kendall2020training} for a proof.

The above formulas were proved in \citet{scellier2023energy} in the case where the state space is $\mathcal{S} = \mathbb{R^N}$. In Appendix~\ref{sec:ep-inequality}, we prove that these formulas still hold when the feasible set is defined by equality and inequality constraints.

\subsection{Implementation Issues in Hardware}

We conclude this Appendix by discussing some challenges related to hardware implementations.

One issue is that the implementation of the nudging proposed in Section~\ref{sec:ep-nonlinear-resistive-network} (Figure~\ref{fig:augmented-circuit}, bottom left) -- using a voltage source, a resistor and a switch in series in every output branch -- does not allow for negative nudging ($\beta < 0$) since a conductance is always non-negative. To achieve negative nudging, a possibility is to use current sources whose current values are set to $i_{jk}=-\beta (v_{jk}-y_{jk})$, where $\beta \in \mathbb{R}$ (Figure ~\ref{fig:augmented-circuit}, bottom right). This approach was proposed in \citet{kendall2020training}. In this method, the free state (inference) is obtained by setting the current sources to zero.

Another implementation issue concerns the contrastive learning rule \eqref{eq:contrastive-rule}, which requires accessing two different states of the system (free and nudge) to extract the weight gradients. In the experimental demonstration by \citet{dillavou2022demonstration,dillavou2023machine}, two copies of the network are coupled to allow access to both network states (free and nudge) simultaneously, enabling full analog training of the system. A caveat is that this coupled learning approach is subject to mismatches between the two copies of the network, which could cause deviations between measured weight gradients and true weight gradients, potentially affecting optimization. Other practical implementations of EP have resorted to using external memory to store the two states before performing the weight updates \citep{yi2023activity,laydevant2024training}, at the cost of increased energy consumption. Other solutions have been proposed to perform the weight updates using a single network, potentially without external memory. \citet{williams2023flexible} propose using only one of the two states, chosen at random at each step, to extract an unbiased (but high variance) estimator of the weight gradients. \citet{anisetti2024frequency} show that the weight gradients can be extracted using the mean and amplitude of the network response to a sinusoidal nudging signal. \citet{falk2023contrastive} show that the weight update can be performed using a non-equilibrium memory (time convolution of trajectory) of the network state, using integral feedback. Finally, a non-contrastive approach called `Agnostic EP' (AEP) was introduced in \citet{scellier2022agnostic}, where the parameter updates are performed through physical dynamics without relying on external measurement and feedback.

\clearpage
\section{Simulation Details}
\label{sec:simulation-details}

We provide the implementation details to reproduce the results of our simulations of DRNs trained by equilibrium propagation (EP) and backpropagation (BP). The code to reproduce the results is available at \url{https://github.com/rain-neuromorphics/energy-based-learning}

\paragraph{MNIST dataset.}
The MNIST dataset of handwritten digits \citep{lecun1998gradient} consists of 60,000 training examples and 10,000 test examples. Each example $x$ in the dataset is a $28 \times 28$ grayscale image, accomonpanied by a label $y \in \left\{ 0, 1, \ldots, 9 \right\}$ indicating the digit represented by the image.

\paragraph{Network architectures.}
We train five deep resistive network (DRN) models termed DRN-1H, DRN-2H, DRN-3H, DRN-XS and DRN-XL. Each DRN has 1568 input units ($2\times28\times28$) and 10 output units corresponding to the ten classes of the MNIST dataset. The DRN-kH model has k hidden layers of 1024 units each ($k \in \{ 1,2,3 \}$). The DRN-XS model has one hidden layer of 100 units. The DRN-XL model has one hidden layer of 32768 units. Table~\ref{table:hyperparams} provides the architectural details of the five DRN models, and the hyperparameters used to obtain the results presented in Table~\ref{table:results}.

\paragraph{Initialization of the conductances (`weights').} For two consecutive layers of sizes $N_{\ell}$ and $N_{\ell+1}$, the matrix of conductances between these layers is initialized using a modified `Kaiming uniform' initialization scheme:
\begin{equation}
g_{ij}^{(\ell+1)} = \max(0,w), \qquad w \sim \mathcal{U}(-c,+c), \qquad c = \sqrt{\frac{1}{N_{\ell}}}.
\end{equation}

\paragraph{Input amplification factor.} We recall that, in a DRN, not only is the number of input nodes (input voltage sources) doubled, but input signals are amplified by a fixed gain factor $A>0$. For each DRN model (DRN-1H, DRN-2H, DRN-3H, DRN-XS and DRN-XL), the choice of $A$ is reported in Table~\ref{table:hyperparams}. We note that the values of $A$ are relatively large, e.g. $A=4000$ for the DRN-3H model. Instead of amplifying input signals of a DRN by a large amplification factor $A \gg 1$, another option is to use bidirectional amplifiers at every layer -- see Appendix~\ref{sec:bidirectional-amplifiers} for details.

\paragraph{Energy minimization.}
To compute the network's steady state, we use our exact block coordinate descent algorithm for DRNs (Section~\ref{sec:block-coordinate-descent}). At every iteration, we first update the layers of even indices (one half of the layers), then we update the layers of odd indices (the other half). We repeat as many iterations as is necessary until convergence to the steady state. In Table~\ref{table:hyperparams}, $T$ denotes the number of iterations performed during inference (free phase), and $K$ denotes the number of iterations performed during the second (training) phase.

\paragraph{EP training procedure.} We trained our networks with equilibrium propagation (EP) and backpropagation (BP). First, we describe EP, as detailed in Appendix~\ref{sec:equilibrium-propagation}. At each training step, we proceed as follows. First we pick a mini-batch of samples from the training set, $x$, and their corresponding labels, $y$. Then we perform $T$ iterations of the block coordinate descent algorithm without nudging ($\beta=0$). We compute the training loss and training error rate for the current mini-batch, to monitor training. We also store the steady state (free state) $v_\star$. Next, we set the nudging parameter to $\beta > 0$ and we perform a new block-coordinate descent of $K$ iterations, to compute the positively-perturbed state $v_\star^\beta$. Then, we reset the network state to the free state $v_\star$, we set the nudging parameter to $-\beta$, and perform a new block-coordinate descent of $K$ iterations to compute the negatively-perturbed state $v_\star^{-\beta}$. Finally, we update all the conductances in proportion to the `centered EP' learning rule.

\paragraph{BP training procedure.} Similar to EP, we first perform $T$ iterations of block coordinate descent to compute the free state. Then we perform another $K$ steps of block coordinate descent (still using $\beta=0$) and backpropagate through these $K$ steps to compute the weight gradients. The procedure is equivalent to performing `truncated backprop through time' where we perform $T+K$ steps in the forward pass, and backpropagate through only the last $K$ steps in the backward pass. This baseline was also used e.g. in \citet{ernoult2019updates} and \citet{scellier2023energy}.

\paragraph{Optimizer and scheduler.}
We use standard mini-batch gradient descent (SGD) with a mini-batch size of 4. No momentum or weight decay is used. We use a scheduler with a learning rate decay of $0.99$ at each training epoch.

\paragraph{Computational resources.}
The code for the simulations uses PyTorch 1.13.1 and TorchVision 0.14.1. \cite{paszke2017automatic}. The simulations were carried out on a single Nvidia A100 GPU.

\begin{table}[ht!]
\caption{Hyper-parameters used for initializing and training the five DRN models (DRN-XS, DRN-XL, DRN-1H, DRN-2H and DRN-3H) to reproduce the results in Table~\ref{table:results}. LR stands for `learning rate'.}
\label{table:hyperparams}
\vskip 0.15in
\begin{center}
\begin{small}
\begin{sc}
\begin{tabular}{cccccc}
\toprule
& DRN-XS & DRN-XL & DRN-1H & DRN-2H & DRN-3H \\
\midrule
Input amplification factor ($A$) & 100 & 800 & 480 & 2000 & 4000 \\
Nudging ($\beta$) & 1.0 & 1.0 & 1.0 & 1.0 & 2.0 \\
Num. iterations at inference ($T$) & 4 & 4 & 4 & 5 & 6 \\
Num. iterations during training ($K$) & 4 & 4 & 4 & 5 & 6 \\
Layer 0 (input) size ($N_0$) & 2-28-28 & 2-28-28 & 2-28-28 & 2-28-28 & 2-28-28 \\
Layer 1 size ($N_1$) & 100 & 32768 & 1024 & 1024 & 1024 \\
Layer 2 size ($N_2$) & 10 & 10 & 10 & 1024 & 1024 \\
Layer 3 size ($N_3$) & &  &  & 10 & 1024 \\
Layer 4 size ($N_4$) & &  &  &  & 10 \\
lr weight 1 \& bias 1 ($\eta_1$) & 0.006 & 0.006 & 0.006 & 0.002 & 0.005 \\
lr weight 2 \& bias 2 ($\eta_2$) & 0.006 & 0.006 & 0.006 & 0.006 & 0.02 \\
lr weight 3 \& bias 3 ($\eta_3$) & &  &  & 0.018 & 0.08 \\
lr weight 4 \& bias 4 ($\eta_4$) & &  &  &  & 0.005 \\
lr decay & 0.99 & 0.99 & 0.99 & 0.99 & 0.99 \\
Mini-batch size & 4 & 4 & 4 & 4 & 4 \\
Number of epochs & 10 & 100 & 50 & 50 & 50 \\
\bottomrule
\end{tabular}
\end{sc}
\end{small}
\end{center}
\vskip -0.1in
\end{table}

\clearpage
\section{Bidirectional Amplifiers}
\label{sec:bidirectional-amplifiers}

One challenge with deep resistive networks (DRNs) is the decay in signal amplitude across the network's layers as depth increases. To compensate for this decay, in our construction, we amplify the input voltages by a large gain factor $A \gg 1$ to compensate for this decay. A solution to this problem, proposed in \citet{kendall2020training}, is to equip each `unit' with a `bidirectional amplifier'. In this appendix, we review this method.

A bidirectional amplifier is a three-terminal device, with bottom ($B$), left ($L$) and right ($R$) terminals. The bottom terminal is linked to ground. The current and voltage states of the left and right terminals, $(v_L,i_L)$ and $(v_R,i_R)$, satisfy the relationship $v_R = a \, v_L$ and $i_R = a \, i_L$, where $a$ is the gain of the bidirectional amplifier. See Figure~\ref{fig:bidirectional-amplifier}. The bidirectional amplifier amplifies the voltages in the forward direction by a factor $a>1$ and amplifies the currents in the backward direction by a factor $1/a$. In practice, a bidirectional amplifier can be created by combining a voltage-controlled voltage source (VCVS) and a current-controlled current source (CCCS), as described in \citet{kendall2020training}.

Next, we equip each unit with a bidirectional amplifier, defining the unit's state as the voltage after amplification (see Figure~\ref{fig:bidirectional-amplifier}). We assume that all units in a given layer $\ell$ use the same gain $a^{(\ell)}$. In this context, Theorem~\ref{thm:convex-qp-formulation} can be restated as follows. For each layer $\ell$ we define
\begin{equation}
c^{(\ell)} := 1 \times a^{(1)} \times a^{(2)} \times \ldots \times a^{(\ell)}.
\end{equation}
The energy function $E$ from Eq.~\eqref{eq:energy-function-drn} becomes
\begin{equation}
E(v) = \frac{1}{2} \sum_{\ell=1}^{L} \sum_{j=1}^{N_{\ell-1}} \sum_{k=1}^{N_\ell} g_{jk}^{(\ell)} \left( \frac{v_j^{(\ell-1)}}{c^{(\ell-1)}} - \frac{v_k^{(\ell)}}{c^{(\ell)}} \right)^2,
\end{equation}
while the feasible set $\mathcal{S}$ from Eq.~\eqref{eq:feasible-set-drn} remains unchanged:
\begin{equation}
\mathcal{S} = \{ v \in \mathbb{R}^{\sum_{\ell=1}^L N_\ell} \mid v_k^{(\ell)} \geq 0 \text{ if k is even}, v_k^{(\ell)} \leq 0 \text{ if k is odd}, 1 \leq \ell \leq L-1, 1 \leq k \leq N_\ell \}.
\end{equation}
The update rule for $v_k^{(\ell)}$ from Eq.~\eqref{eq:update-rule-hidden-drn} becomes:
\begin{equation}
v_k^{(\ell)} \leftarrow 
\left\{
\begin{array}{l}
\max \left( 0, p_k^{(\ell)} \right) \qquad \text{if k is even}, \\
\min \left( 0, p_k^{(\ell)} \right) \qquad \text{if k is odd},
\end{array}
\right.
\qquad \text{where} \qquad
p_k^{(\ell)} := \frac{\sum_{j=1}^{N_{\ell-1}} g_{jk}^{(\ell)} a^{(\ell)} v_j^{(\ell-1)} + \sum_{j=1}^{N_{\ell+1}} g_{kj}^{(\ell+1)} \frac{v_j^{(\ell+1)}}{a^{(\ell+1)}}}{\sum_{j=1}^{N_{\ell-1}} g_{jk}^{(\ell)} + \sum_{j=1}^{N_{\ell+1}} g_{kj}^{(\ell+1)}}.
\end{equation}
The update rule for the output layer from Eq.~\eqref{eq:update-rule-output-drn} becomes:
\begin{equation}
v_k^{(L)} \leftarrow 
\frac{\sum_{j=1}^{N_{\ell-1}} g_{jk}^{(L)}  a^{(L)} v_j^{(L-1)} + \beta y_k}{\sum_{j=1}^{N_{\ell-1}} g_{jk}^{(L)} + \beta}.
\end{equation}
The learning rule for the conductances of a DRN with bidirectional amplifiers also needs slight modification. For the conductance between unit $j$ of layer $\ell-1$ and unit $k$ of layer $\ell$, we denote:
\begin{equation}
g_{jk} := g_{jk}^{(\ell)}, \qquad \widetilde{v}_j := \frac{v_j^{(\ell-1)}}{c^{(\ell-1)}}, \qquad \widetilde{v}_k := \frac{v_k^{(\ell)}}{c^{(\ell)}}.
\end{equation}
Then the learning rule reads:
\begin{equation}
\Delta g_{jk} = \frac{\eta}{2\beta} \left[ \left( \widetilde{v}_j^0 - \widetilde{v}_k^0 \right)^2 -  \left( \widetilde{v}_j^\beta - \widetilde{v}_k^\beta \right)^2 \right],
\end{equation}
where $\widetilde{v}^0$ and $\widetilde{v}^\beta$ denote the free state and nudged state values, respectively. If bidirectional amplifiers are removed from the network (choosing $a^{(\ell)}=1$ for all $\ell$), then $c^{(\ell)}=1$ for all $\ell$, and the original formulas from Section~\ref{sec:deep-resistive-network} and Appendix~\ref{sec:equilibrium-propagation} are recovered.

\begin{figure}
\begin{center}
\fbox{
\includegraphics[width=0.4\textwidth]{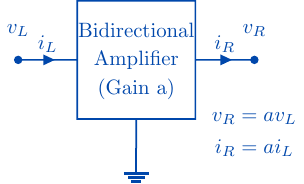}
}
\vspace{0.5cm}
\includegraphics[width=0.95\textwidth]{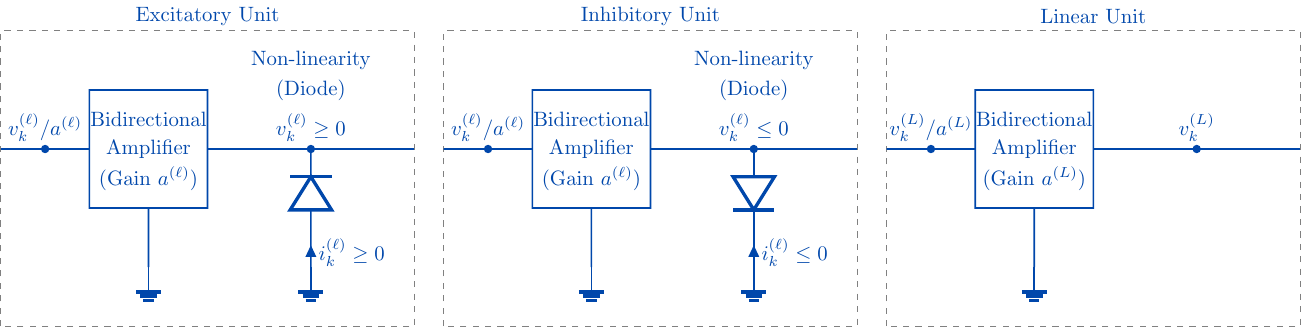}
\fbox{
\includegraphics[width=0.95\textwidth]{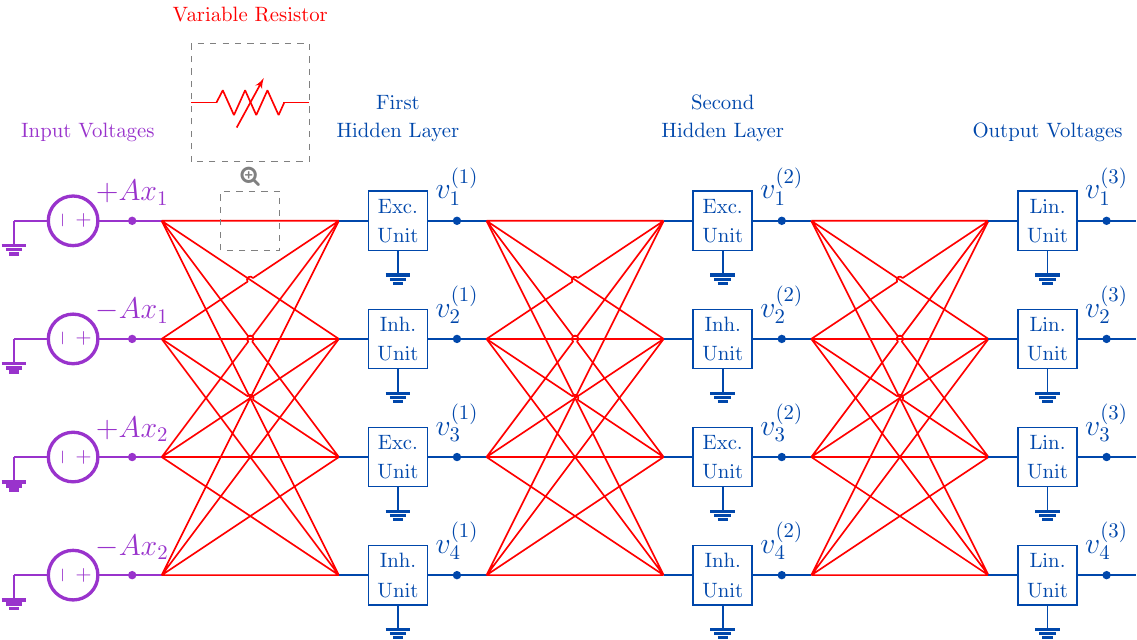}
}
\end{center}
\caption{
\textbf{Top.} Bidirectional amplifier with gain $a$. The right terminal voltage ($v_R$) is related to the left terminal voltage ($v_L$) by $v_R = a v_L$, where $a$ is a gain factor. The left terminal current ($i_L$) is related to the right terminal current ($i_R$) by $i_L = \frac{1}{a} i_R$.
\textbf{Middle.} A unit consisting of a bidirectional amplifier, possibly followed by a diode between the unit's node and ground. Units come in two types: excitatory and inhibitory, depending on the diode's orientation.
\textbf{Bottom.} A deep resistive network with bidirectional amplifiers.
\label{fig:bidirectional-amplifier}
}
\end{figure}

\clearpage
\section{Continuous Hopfield Networks}
\label{sec:continuous-hopfield-networks}

In this appendix, we discuss the similarities and differences between resistive networks and continuous Hopfield networks (CHNs). First, we present the CHN model and its layered version, the deep Hopfield network (DHN). Then, we show that nonlinear resistive networks can be cast as CHNs. Finally, we compare DRNs and DHNs through simulations.

\subsection{Continuous Hopfield Networks and Deep Hopfield Networks}

The continuous Hopfield network (CHN), introduced by \citet{hopfield1984neurons}, is inspired by the Ising model of spin glasses. A CHN consists of continuous-valued units ($s_k$) interconnected by bidirectional weights ($w_{jk}$). Each unit $k$ has an associated bias ($b_k$), and a nonlinear activation function $f_k : \mathbb{R} \to \mathbb{R}$, assumed to be increasing. The network's state is represented by the vector $s=(s_1,s_2,\ldots,s_N)$, where $N$ is the number of units. The energy function of the network, called the `Hopfield energy function', is defined as
\begin{equation}
\label{eq:hopfield-energy}
E_{\rm Hopfield}(s_1,\ldots,s_N) := \sum_k \underbrace{\int_0^{s_k} f_k^{-1}(u)du}_{:= \Phi_k(s_k)} - \sum_{j < k} w_{jk} s_j s_k - \sum_k b_k s_k.
\end{equation}
Since the nonlinear activation function $f_k$ is increasing, its inverse function $f_k^{-1}$ is also increasing, making $\Phi_k$ convex. This implies that the function $s_k \mapsto E_{\rm Hopfield}(s_1,\ldots,s_N)$ is convex, given the state of other units fixed. The value of $s_k$ that minimizes this function satisfies
\begin{align}
\partial_{s_k} E_{\rm Hopfield}(s_1,\ldots,s_N) = 0 \quad & \Longleftrightarrow \quad f_k^{-1}(s_k) - \sum_{j \neq k} w_{jk} s_j - b_k = 0 \\
& \Longleftrightarrow \quad s_k = f_k \left( \sum_{j \neq k} w_{jk} s_j + b_k \right).
\label{eq:hopfield-update}
\end{align}
Thus, the Hopfield energy function can be minimized via exact coordinate descent, where each step updates one unit $k$ according to \eqref{eq:hopfield-update} -- similar to our algorithm for nonlinear resistive network (Theorem~\ref{thm:coordinate-descent}).

In the CHN model of \citet{scellier2017equilibrium} where equilibrium propagation was first illustrated, the nonlinear activation function is the `hard sigmoid'
\begin{equation}
\label{eq:hard-sigmoid}
f_k(s_k) := \min(\max(0,s_k),1), \qquad 1 \leq k \leq N.
\end{equation}
The term $\Phi_k(s_k)$ in the Hopfield energy function \eqref{eq:hopfield-energy} becomes
\begin{equation}
\Phi_k(s_k) := \int_0^{s_k} f_k^{-1}(u)du = \left\{
\begin{array}{ll}
+\infty & \text{if} \quad s_k \leq 0, \\
s_k^2   & \text{if} \quad 0 \leq s_k \leq 1, \\
+\infty & \text{if} \quad 1 \leq s_k.
\end{array}
\right.
\end{equation}
In this context, the steady state of the CHN can be cast as the solution of the following constrained minimization problem:
\begin{equation}
s_\star = \underset{s \in \mathcal{S}_{\rm Hopfield}}{\arg \min} \; E_{\rm Hopfield}^{\rm ideal}(s),
\end{equation}
where $E_{\rm Hopfield}^{\rm ideal}$ is the Hopfield energy function corresponding to the hard sigmoid activation function \eqref{eq:hard-sigmoid}, and $\mathcal{S}_{\rm Hopfield}^{\rm ideal}$ is the set of feasible configurations:
\begin{gather}
\label{eq:ideal-hopfield-energy}
E_{\rm Hopfield}^{\rm ideal}(s_1, \ldots, s_N) := \frac{1}{2} \sum_k  s_k^2 - \sum_{(j,k)} w_{jk} s_j s_k - \sum_k b_k s_k, \\
\label{eq:feasible-set-chn}
\mathcal{S}_{\rm Hopfield}^{\rm ideal} := \{ s=(s_1,s_2,\ldots,s_N) \in \mathbb{R}^N \mid 0 \leq s_k \leq 1, \quad 1 \leq k \leq N \}.
\end{gather}
Although the Hopfield energy function is convex in $s_k$ (given the other units fixed), it is not necessarily convex in the overall state $s = (s_1, s_2, \ldots s_N)$. This non-convexity leads to the multi-stability issue (discussed by \citet{wang2024training} in the context of coupled phase oscillators). In contrast, nonlinear resistive networks have a convex energy function, resulting in a unique steady state (Theorem~\ref{thm:convex-qp-formulation}).

The units of a CHN can be organised into layers to form a `deep Hopfield network' (DHN). Similar to the deep resistive network (DRN), the bipartite structure of a DHN allows for a block coordinate descent method, where half of the units are updated in parallel at each step -- either the layers of odd index or the layers of even index. This method, used in the simulations of \citet{scellier2023energy}, is akin to the block Gibbs sampling procedure for deep Boltzmann machines \citep{salakhutdinov2009deep}.

\subsection{Casting Resistive Networks as Continuous Hopfiled Networks}

Next we show that a nonlinear resistive network can be cast as a CHN. First we recall the energy function of a nonlinear resistive network (Theorem~\ref{thm:convex-qp-formulation}):
\begin{equation}
E_{\rm resistive} := \frac{1}{2} \sum_{(j,k) \in \mathcal{B}_{\rm R}} g_{jk} \left( v_j - v_k \right)^2 + \sum_{(j,k) \in \mathcal{B}_{\rm CS}} i_{jk}^{\rm CS} \left( v_j - v_k \right).
\end{equation}
We define $G_j$ as the sum of conductances linked to node $j$, and $I_j$ as the sum of currents flowing to node $j$ through current sources:
\begin{equation}
G_j := \sum_{k: (j,k) \in \mathcal{B}_{\rm R}} g_{jk}, \qquad I_j := \sum_{k: (j,k) \in \mathcal{B}_{\rm CS}} i_{jk}^{\rm CS}.
\end{equation}
Then, we perform a change of variables:
\begin{equation}
s_j := \sqrt{G_j} v_j, \qquad w_{jk} := \frac{g_{jk}}{\sqrt{G_j G_k}}, \qquad b_j := \frac{I_j}{\sqrt{G_j}},
\end{equation}
so the energy function of the resistive network rewrites in terms of the new variables:
\begin{equation}
E_{\rm resistive} = \frac{1}{2} \sum_j s_j^2 - \sum_{j,k} w_{jk} s_j s_k - \sum_k b_k s_k.
\end{equation}
This is the energy function of an ideal CHN (Eq.~\eqref{eq:ideal-hopfield-energy}). Moreover, the constraints on the electrical potentials ($v_k$) imposed by the diodes and voltage sources (Eq.~\eqref{eq:feasible-set}) translate into equality and inequality constraints on the new variables ($s_k$), similar (though not identical) to those of the feasible set of the CHN (Eq.~\eqref{eq:feasible-set-chn}). This shows that (ideal) nonlinear resistive networks are, in some sense, a subclass of (ideal) CHNs.

The converse statement is not true: a CHN with a non-convex energy function cannot be converted into an electrical network composed solely of resistors, diodes, voltage sources and current sources. Such electrical networks have a convex energy function (Theorem~\ref{thm:convex-qp-formulation}).

\subsection{Comparison of DRNs and DHNs in Simulations}

Finally, we compare DRNs and deep Hopfield networks (DHNs) of the same size through simulations. Using equilibrium propagation, we train DHNs of one, two and three hidden layers (each comprising 1024 units), denoted DHN-1H, DHN-2H and DHN-3H, respectively, on the MNIST dataset. Similar to the DRN simulations in Table~\ref{table:results}, we use exact block coordinate descent to compute the steady state of the DHN. This contrasts with the DHN simulations of \citet{scellier2017equilibrium} where the steady state was obtained using gradient descent in the feasible set.

\begin{table}[ht!]
\caption{Comparison of deep Hopfield network (DHN) and deep resistive network (DRN) architectures on MNIST using equilibrium propagation (EP) for training. Test error rates (in \%) are reported as mean $\pm$ standard deviation over five runs. We also report the number of iterations of exact block coordinate descent to compute the free state ($T$) and the nudge state ($K$).}
\label{table:hopfield-networks}
\vskip 0.15in
\begin{center}
\begin{small}
\begin{sc}
\begin{tabular}{ccccccc}
\toprule
& \multicolumn{3}{c}{DRN} & \multicolumn{3}{c}{DHN} \\
\cmidrule(r){2-4} \cmidrule(r){5-7}
& T & K & Test error (\%) & T & K & Test error (\%) \\
\midrule
1h & 4 & 4 & 1.57 $\pm$ 0.07 & 15 & 15 & 1.79 $\pm$ 0.06 \\
2h & 5 & 5 & 1.48 $\pm$ 0.05 & 50 & 20 & 1.65 $\pm$ 0.05 \\
3h & 6 & 6 & 1.66 $\pm$ 0.09 & 100 & 10 & 1.65 $\pm$ 0.02 \\
\bottomrule
\end{tabular}
\end{sc}
\end{small}
\end{center}
\vskip -0.1in
\end{table}

Table~\ref{table:hopfield-networks} shows the results. We observe that DRNs require fewer iterations of exact block coordinate descent than their equivalent-size DHNs to converge to a steady state. This can be explained by their convex energy function, while DHNs may have non-convex energy functions. We also observe that DRNs achieve similar (or even slightly better) performance compared to DHNs. This indicates that the convexity of the energy function in DRNs does not limit their computational expressivity. In fact, \citet{scellier2023universal} have proven a universal approximation theorem for DRNs.

\clearpage
\section{Equilibrium Propagation with Inequality Constraints}
\label{sec:ep-inequality}

The equilibrium propagation (EP) formulas detailed in Appendix~\ref{sec:ep-formulas} were proved in \citet{scellier2017equilibrium,scellier2023energy}. These derivations assumed that the equilibrium state $s_\star$ satisfies the stationary condition $\frac{\partial E}{\partial s}(s_\star) = 0$. However, this condition may not hold when the feasible set is defined by inequality constraints, as seen in the settings of nonlinear resistive networks (Section~\ref{sec:nonlinear-resistive-network}) and continuous Hopfield networks (Appendix~\ref{sec:continuous-hopfield-networks}). In these settings, the stationary condition is satisfied when the equilibrium state $s_\star$ is within the interior of the feasible region $\mathcal{S}$, but it may not hold when $s_\star$ lies on the boundary of $\mathcal{S}$. In this appendix, we demonstrate that the EP formulas remain valid even when the feasible set is defined by inequality constraints.

Let $\Theta := \mathbb{R}^m$ denote the parameter space, and
\begin{equation}
\label{eq:feasible-set-general}
\mathcal{S} := \{ s \in \mathbb{R}^n \mid f_1(s) \leq 0, \ldots, f_K(s) \leq 0 \}
\end{equation}
denote the feasible set, where $f_k: \mathbb{R}^n \to \mathbb{R}$ is a smooth function for each $k$. Given two functionals $C: \Theta \times \mathcal{S} \to \mathbb{R}$ and $E: \Theta \times \mathcal{S} \to \mathbb{R}$, called `cost function' and `energy function', respectively, we consider the following bilevel optimization problem \citep{zucchet2022beyond}:
\begin{gather}
\text{find} \; \min_{\theta \in \Theta} C(\theta,s(\theta)), \\
\text{subject to } s(\theta) := \underset{s \in \mathcal{S}}{\arg \min} \; E(\theta,s).
\label{eq:free-state}
\end{gather}
Using the notations from Appendix~\ref{sec:equilibrium-propagation}, we further define
\begin{gather}
F(\beta,\theta,s) := E(\theta,s) + \beta C(\theta,s), \\
\label{eq:nudge-state}
s_\theta^\beta := \underset{s \in \mathcal{S}}{\arg\min} \; F(\beta,\theta,s), \\
G(\beta,\theta) := F \left( \beta,\theta,s_\theta^\beta \right), \\
\mathcal{L}_\beta \left( \theta \right) := \frac{G(\beta,\theta) - G(0,\theta)}{\beta}.
\end{gather}
We refer to $\mathcal{L}_\beta$ as the contrastive function.

\begin{thm}[Equilibrium propagation formulas]
\label{thm:ep-inequality}
The gradient of the contrastive function $\mathcal{L}_\beta$ is given by:
\begin{equation}
\nabla_\theta \mathcal{L}_\beta(\theta) = \frac{1}{\beta} \left( \frac{\partial F}{\partial \theta} \left( \beta,\theta,s_\theta^\beta \right) - \frac{\partial F}{\partial \theta} \left( 0,\theta,s^0 \right) \right).
\end{equation}
Furthermore, the contrastive function satisfies:
\begin{equation}
\forall \beta \in \mathbb{R}, \quad \mathcal{L}_\beta(\theta) = C \left( \theta,s(\theta) \right) + O(\beta), \qquad \text{and} \qquad \forall \beta>0, \quad \mathcal{L}_\beta(\theta) \leq C \left( \theta,s(\theta) \right) \leq \mathcal{L}_{-\beta}(\theta).
\end{equation}
\end{thm}

We prove Theorem~\ref{thm:ep-inequality} by extending the proofs from \citet{scellier2022agnostic,scellier2023energy}. The new addition is the introduction of the Lagrangian $L$ and the Lagrange multipliers associated with the constrained minimization problems of Eq.~\eqref{eq:free-state} and Eq.~\eqref{eq:nudge-state}, which brings us back to the setting where the stationary condition is satisfied.

\begin{proof}[Proof of Theorem~\ref{thm:ep-inequality}]
Since $\mathcal{S}$ is defined by inequality constraints, Eq.~\eqref{eq:nudge-state} represents constrained optimization problem. We introduce the corresponding Lagrangian function:
\begin{equation}
L \left( \beta,\theta,s,\lambda \right) := F(\beta,\theta,s) + \sum_k \lambda_k f_k(s),
\end{equation}
where the functions $f_k$ define the feasible set (Eq.~\eqref{eq:feasible-set-general}), and $\lambda = (\lambda_1, \cdots, \lambda_K) \in \mathbb{R}^K$ are the Lagrange multipliers. Since $s_\theta^\beta$ is a minimum of this constrained optimization problem, there exist Lagrange multipliers $\lambda_\theta^\beta = (\lambda_1^\beta, \cdots, \lambda_K^\beta)$ such that $(s_\theta^\beta,\lambda_\theta^\beta)$ satisfies the Karush-Kuhn-Tucker (KKT) conditions. The stationarity condition for the Lagrangian is
\begin{equation}
\label{eq:stationarity-general}
\frac{\partial L}{\partial s} \left( \beta,\theta,s_\theta^\beta,\lambda_\theta^\beta \right) = 0,
\end{equation}
while the primal feasibility condition, dual feasibility condition, and the complementary slackness for inequality constraints are:
\begin{equation}
\label{eq:KKT-conditions}
f_k(s_\theta^\beta) \leq 0, \qquad \lambda_k^\beta \geq 0, \qquad \lambda_k^\beta f_k(s_\theta^\beta) = 0, \qquad 1 \leq k \leq K.
\end{equation}
Next, we calculate the gradient of the contrastive function:
\begin{equation}
\nabla_\theta \mathcal{L}_\beta(\theta) = \frac{1}{\beta} \left( \frac{\partial G}{\partial \theta}(\beta,\theta) - \frac{\partial G}{\partial \theta}(0,\theta) \right).
\end{equation}
Using the definitions of $G$ and $L$, and the complementary slackness condition, we have
\begin{equation}
G(\beta,\theta) = L \left( \beta,\theta,s_\theta^\beta,\lambda_\theta^\beta \right).
\end{equation}
Therefore, applying the chain rule of differentiation, we get
\begin{align}
\frac{\partial G}{\partial \theta} \left( \beta,\theta \right) & = \frac{\partial L}{\partial \theta} \left( \beta,\theta,s_\theta^\beta,\lambda_\theta^\beta \right) + \underbrace{\frac{\partial L}{\partial s} \left( \beta,\theta,s_\theta^\beta,\lambda_\theta^\beta \right)}_{=0} \cdot \frac{\partial s_\theta^\beta}{\partial \theta} + \sum_k \frac{\partial L}{\partial \lambda_k} \left( \beta,\theta,s_\theta^\beta,\lambda_\theta^\beta \right) \cdot \frac{\partial \lambda_k^\beta}{\partial \theta} \\
& = \frac{\partial F}{\partial \theta} \left( \beta,\theta,s_\theta^\beta \right) + \sum_k f_k(s_\theta^\beta) \cdot \frac{\partial \lambda_k^\beta}{\partial \theta}
\end{align}
Here we have used the stationarity condition of the Lagrangian, Eq.~\eqref{eq:stationarity-general}. We claim that
\begin{equation}
\label{eq:claim}
f_k(s_\theta^\beta) \cdot \frac{\partial \lambda_k^\beta}{\partial \theta} = 0, \qquad 1 \leq k \leq K,
\end{equation}
which we prove next. Let $k$ be fixed. The primal feasibility condition of Eq.~\eqref{eq:KKT-conditions} imposes that $f_k(s_\theta^\beta) \leq 0$. If $f_k(s_\theta^\beta) = 0$, then Eq.~\eqref{eq:claim} holds. Otherwise, $f_k(s_\theta^\beta) < 0$. Using the smoothness of $f_k$ and the smoothness of $\theta^\prime \mapsto s_{\theta^\prime}^\beta$ (guaranteed by the implicit function theorem, assuming that $F$ is smooth too), there exists a neighborhood of $\theta$ such that $f_k(s_{\theta^\prime}^\beta) < 0$ for every $\theta^\prime$ in this neighborhood. By the complementary slackness of Eq.~\eqref{eq:KKT-conditions}, it follows that $\lambda_k^\beta = 0$ for every $\theta^\prime$ in this neighborhood, therefore $\frac{\partial \lambda_k^\beta}{\partial \theta} = 0$. Hence, Eq.~\eqref{eq:claim} holds again. As a consequence
\begin{equation}
\frac{\partial G}{\partial \theta} \left( \beta,\theta \right) = \frac{\partial F}{\partial \theta} \left( \beta,\theta,s_\theta^\beta \right).
\end{equation}
Hence,
\begin{equation}
\nabla_\theta \mathcal{L}_\beta(\theta) = \frac{1}{\beta} \left( \frac{\partial F}{\partial \theta} \left( \beta,\theta,s_\theta^\beta \right) - \frac{\partial F}{\partial \theta} \left( 0,\theta,s_\theta^0 \right) \right),
\end{equation}
which is the first half of Theorem~\ref{thm:ep-inequality}.

Next, we prove the second half: the properties of the contrastive function. Using again the chain rule of differentiation, we have
\begin{align}
\frac{\partial G}{\partial \beta} \left( \beta,\theta \right) & = \frac{\partial L}{\partial \beta} \left( \beta,\theta,s_\theta^\beta,\lambda_\theta^\beta \right) + \underbrace{\frac{\partial L}{\partial s} \left( \beta,\theta,s_\theta^\beta,\lambda_\theta^\beta \right)}_{=0} \cdot \frac{\partial s_\theta^\beta}{\partial \beta} + \sum_k \frac{\partial L}{\partial \lambda_k} \left( \beta,\theta,s_\theta^\beta,\lambda_\theta^\beta \right) \cdot \frac{\partial \lambda_k^\beta}{\partial \beta} \\
& = C \left( \theta,s_\theta^\beta \right) + \sum_k \underbrace{f_k(s_\theta^\beta) \cdot \frac{\partial \lambda_k^\beta}{\partial \beta}}_{=0} = C \left( \theta,s_\theta^\beta \right).
\end{align}
Here again, we have used the stationary condition of the Lagrangian, and the fact that
\begin{equation}
f_k(s_\theta^\beta) \frac{\partial \lambda_k^\beta}{\partial \beta} = 0, \qquad 1 \leq k \leq K,
\end{equation}
which we prove similarly as before. Evaluating the above expression at $\beta=0$, we get
\begin{equation}
\frac{\partial G}{\partial \beta} \left( 0,\theta \right) = C \left( \theta,s(\theta) \right),
\end{equation}
since $s_\theta^0 = s(\theta)$ by definition. Using a Taylor expansion of $G(\beta,\theta)$ around $\beta=0$, we get
\begin{equation}
G(\beta,\theta) = G(0,\theta) + \beta C \left( \theta,s(\theta) \right) + O(\beta^2).
\end{equation}
Subtracting $G(\theta,0)$ on both sides and dividing by $\beta$, we get
\begin{equation}
\mathcal{L}_\beta(\theta) = \frac{G(\beta,\theta) - G(0,\theta)}{\beta} = C \left( \theta,s(\theta) \right) + O(\beta).
\end{equation}
As for the upper bound and lower bound properties, we write for any $\beta \in \mathbb{R}$,
\begin{equation}
G(\beta,\theta) = \inf_{s \in \mathcal{S}} \left\{E(\theta,s)+ \beta C(\theta,s)\right\} \leq E \left( \theta,s(\theta) \right) + \beta C \left( \theta,s(\theta) \right) = G(0,\theta) + \beta C \left( \theta,s(\theta) \right).
\end{equation}
Subtracting $G(0,\theta)$ on both sides we get
\begin{equation}
G(\beta,\theta) - G(0,\theta) \leq \beta C(\theta,s(\theta)).
\end{equation}
Next we divide by $\beta$. For $\beta > 0$, we get
\begin{equation}
\mathcal{L}_\beta(\theta) = \frac{G(\beta,\theta) - G(0,\theta)}{\beta} \leq C(\theta,s(\theta)), \qquad \beta > 0,
\end{equation}
and for $\beta < 0$ we get
\begin{equation}
\mathcal{L}_\beta(\theta) = \frac{G(\beta,\theta) - G(0,\theta)}{\beta} \geq C(\theta,s(\theta)), \qquad \beta < 0.
\end{equation}
Hence the result.
\end{proof}

\end{document}